\documentclass{article}

\usepackage{arxiv}
\usepackage{algorithm}% http://ctan.org/pkg/algorithms
\usepackage{algpseudocode}% http://ctan.org/pkg/algorithmicx
\usepackage[utf8]{inputenc} % allow utf-8 input
\usepackage[T1]{fontenc}    % use 8-bit T1 fonts
\usepackage{hyperref}       % hyperlinks
\usepackage{color}
\usepackage{amsthm}
\usepackage{amsmath}
\usepackage{url}            % simple URL typesetting
\usepackage{booktabs}       % professional-quality tables
\usepackage{amsfonts}       % blackboard math symbols
\usepackage{nicefrac}       % compact symbols for 1/2, etc.
\usepackage{microtype}      % microtypography
\usepackage{lipsum}		% Can be removed after putting your text content
\usepackage{graphicx}
\usepackage{natbib}
\usepackage{doi}

\newtheorem{thm}{Proposition}

\title{Fractional Polynomials Models as Special Cases of Bayesian Generalized Nonlinear Models}

\author{ \href{https://orcid.org/0000-0002-3244-6571}{\includegraphics[scale=0.06]{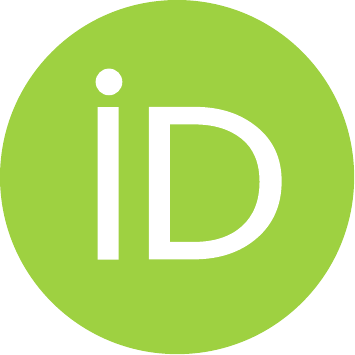}\hspace{1mm}Aliaksandr Hubin}\thanks{corresponding author} \\
	Department of Mathematics,\\ University of Oslo\\
	BIAS, NMBU\\OUC \\
	\texttt{aliaksandr.hubin@nmbu.no} \\
	\And
	\href{https://orcid.org/0000-0000-0000-0000}{\includegraphics[scale=0.06]{orcid.pdf}\hspace{1mm}Georg Heinze} \\
	 Institute of Clinical Biometrics,\\Center for Medical Data Science,\\ Medical University of Vienna\\
	\texttt{georg.heinze@meduniwien.ac.at} \\
	\And
	\href{https://orcid.org/0000-0002-7441-6880}{\includegraphics[scale=0.06]{orcid.pdf}\hspace{1mm}Riccardo De Bin} \\
	Department of Mathematics,\\ University of Oslo\\
	\texttt{debin@math.uio.no} \\
}

\renewcommand{\shorttitle}{\textit{arXiv} Template}

\hypersetup{
pdftitle={A template for the arxiv style},
pdfsubject={q-bio.NC, q-bio.QM},
pdfauthor={David S.~Hippocampus, Elias D.~Striatum},
pdfkeywords={First keyword, Second keyword, More},
}

\begin{document}
\maketitle

% Brief abstract of the paper:
\begin{abstract}
We propose a framework for fitting fractional polynomials models as special cases of Bayesian Generalized Nonlinear Models, applying an adapted version of the Genetically Modified Mode Jumping Markov Chain Monte Carlo algorithm. The universality of the Bayesian Generalized Nonlinear Models allows us to employ a Bayesian version of the fractional polynomials models in any supervised learning task, including regression, classification, and time-to-event data analysis. We show through a simulation study that our novel approach performs similarly to the classical frequentist fractional polynomials approach in terms of variable selection, identification of the true functional forms, and prediction ability, while providing, in contrast to its frequentist version, a coherent inference framework. Real data examples provide further evidence in favor of our approach and show its flexibility.
\end{abstract}

% Keywords (at most 5):

\section{Introduction}
Linear regression models are arguably among the most popular tools in statistics, especially  their generalized versions (GLM). As the name suggests, they model a (function of a) response variable with a linear function of the predictors. While imposing a linear structure has many advantages, including the reduction of the variance, often it may not adequately reflect the real effect of the predictor on the response and may lead to nonlinear structures of the residuals, which indicates a violation of the model assumptions and inappropriateness of the standard asymptotic inference procedures. For example, misspecifying a truly non-linear structure in the predictor-response relationship as linear may result in biased estimates of the regression coefficients, a non-constant variance, and finally in a wrong interpretation of the modeling results. Heteroscedasticity can still be a problem in Bayesian linear regression models. The reason is that the posterior distributions of the regression coefficients depend on the likelihood, which in turn depends on the residuals. Residuals with non-constant variance may affect the shape of the likelihood and lead to incorrect posteriors for the regression coefficients. 

Non-linearities in the predictor-response relationship can be adequately captured by flexible modeling approaches like splines, often used within the (generalized) additive model framework. Although powerful and effective, these approaches have the strong drawback of making the model interpretation hard. Roughly speaking, these methods do not supply regression coefficients that can be easily interpreted. For this reason, it is often convenient to transform the predictors with specific global functions, for example by taking the logarithm or the square root, and then assuming a linear relationship of the transformed predictor with the response variable. In a linear model, the corresponding regression coefficient will then have the familiar interpretation of ``expected difference in the response variable for a unit difference of the -- now transformed -- predictor''. Following this way of thinking, \cite{RoystonAltman1994} introduced the fractional polynomial approach. The basic idea is to select the transformation of the predictor among a set of 8 possible functions ($x^{-2}, x^{-1}, x^{-0.5},  \log x, x^{0.5}, x^{1}, x^{2}, x^{3}$), that is then added to the linear model. This set corresponds to a set of powers of polynomials $\{-2,-1,-0.5,0,0.5,2,3\}$ for the Box-Tidwell transformation, where $x^0 = \log(x)$ \citep{BoxTidwell1962}.

Many refinements have been considered, including combinations of these functions \cite[fractional polynomials of order $d$, see][]{RoystonAltman1994}, a multivariable approach \citep{SauerbreiRoyston1999}, a modification to account for interactions \citep{RoystonSauerbrei2004, RoystonSauerbrei2008}, and others. In particular, the fractional polynomials of order $d$, hereafter $FP(d)$, allow multiple transformations of the predictor, such that, for a simple linear model, 
$$
E[Y|X] = \beta_0 + \beta_1 X^{p_1} + \dots + \beta_d X^{p_d}
$$
where $p_1, \dots, p_d$ belong to $\{-2,-1,-0.5,0,0.5,2,3\}$. By convention, the case that $p_{j+1} = p_j$ indicates repeated power with transformations $X^{p_j}, X^{p_j} \log(X)$. While in theory, fractional polynomials of any order $d$ are possible, in practice almost only fractional polynomials of order 1 or 2 have typically been used \citep[][Ch 5.9]{RoystonSauerbrei2008}.

Multivariable fractional polynomials are the natural extension of the procedure to multivariable regression problems. In this case, each predictor $X_1, \dots, X_J$ gets a specific transformation among those allowed by the order of the fractional polynomial. While conceptually straightforward, this modification complicates the fitting procedure, due to the high complexity of the model space. \cite{SauerbreiRoyston1999} proposed a sort of back-fitting algorithm to fit multivariable fractional polynomial (MFP) models. Herein, all variables are first ordered based on the significance of their linear effect (increasing p-values for the hypothesis of no effect). Then, variable by variable, a function selection procedure (FSP) based on a closed testing procedure with likelihood ratio tests is used to decide whether the variable must be included or can be omitted and if it should be included with the best-fitting second-order fractional polynomial, with the best-fitting first-order fractional polynomial or without transformation. The FSP is performed for all predictors, keeping the remaining fixed (as transformed in the previous step) for a pre-specified number of rounds or until there are no differences with the results of the previous cycle.

Limited to Gaussian linear regression, \cite{SabanesboveHeld2011} implemented an alternative approach to MFP under the Bayesian paradigm. Based on hyper-$g$ priors \citep{LiangAl2008}, their procedure explores the model space by MCMC and provides a framework in which inferential results are not affected, for example, by the repeated implementation of likelihood-ratio tests. Obviously, the restriction to Gaussian linear regression problems highly limits the applicability of this procedure. Moreover, while computationally attractive, the MCMC algorithm may struggle to efficiently explore the complicated model space induced by highly correlated predictors.

To address these drawbacks, here we propose a novel approach based on the characterization of the fractional polynomial models as special cases of Bayesian Generalised Nonlinear Models \citep[BGNLM][]{hubin2021flexible} and the implementation of a fitting algorithm based on the Genetically Modified Mode Jumping MCMC (GMJMCMC) algorithm of \cite{hubin2020novel}. Bayesian Generalised Nonlinear Models provide a very general framework that allows us a straightforward implementation of fractional polynomials beyond the linear Gaussian regression case, including, but not limited to, generalized linear models, generalized linear mixed models, Cox regression, and models with interactions. In addition, GMJMCMC enables search through the model space to find the set of models with probability mass. GMJMCMC is an improved version of the Mode Jumping MCMC, which is an MCMC variant designed to better explore the posterior distributions, particularly useful when the posterior distribution is complex and has multiple modes. Mode Jumping struggles when the dimensionality increases, which can be an issue in the case of fractional polynomials, as one needs to explore $2^{8J}$  models (in the case of the first-order fractional polynomials described above) as compared to $2^{J}$ for the linear models. The GMJMCMC algorithm resolves these limitations by creating a genetic evolution of sets of features controlled by the genetic component, where simple Mode Jumping MCMCs can be run. In this work, therefore, we provide a powerful tool for fitting fractional polynomials in various applications.

The rest of the paper is organized as follows. Section \ref{sec:BMFP} describes the multivariable Bayesian fractional polynomials models in the framework of BGNLM, including the fitting algorithm based on GMJMCMC. In Section \ref{sec:ART} the performance of our procedure is evaluated via simulations, while three applications to real data are reported in Section \ref{sec:real}, where we show our approach applied to a regression, classification, and time-to-event data problems. Finally, some remarks conclude the paper in Section \ref{sec:conclusions}.

\section{Methods}\label{sec:BMFP}

\subsection{Bayesian Generalized Nonlinear Models}

Consider the situation with a response variable $Y$ and a $J$-dimensional random vector of input predictors $\mathbf{X} = (X_1, \dots, X_J)$. The Bayesian fractional polynomials models can be seen as special cases of the Bayesian Generalized Nonlinear Models \citep{hubin2021flexible},

\begin{equation}\label{eq:BGNLM}
  \begin{array}{rcl}
  Y &\sim& \mathfrak{f}\left(y | \mu,\phi\right)\\
  h\left(\mu(\boldsymbol X)\right)&=&\alpha+\sum_{j=1}^{m} \gamma_{j}\beta_{j} F_j(\mathbf{X},\eta_j)
  \end{array}
\end{equation} 

where $\mathfrak{f}$ denotes the parametric distribution of $Y$ belonging to the exponential family with mean $\mu$ and dispersion parameter $\phi$. The function $h$ is a link function, $\alpha$ and $\beta_{j}, j = 1, \dots, m$ are unknown parameters, and $\gamma_{j}$ is an indicator variable that specifies whether the (possibly nonlinear) transformation of predictors $F_j$, and its set of inner parameters $\eta_j$, is included in the model.

Equation \eqref{eq:BGNLM} provides a very general framework, that contains as special cases models that span from the linear Gaussian regression to Neural Networks. It is convenient, for our purpose, to constrain the framework to only allow univariate transformations $\rho_{k}(x_j)$, $k = 1, \dots, K$, of the predictors, and regression on the mean, $\mu = \mu(\boldsymbol X)$:

\begin{equation}\label{eq:LRM}
  \begin{array}{rcl}
 Y &\sim& \mathfrak{f}\left(y | \mu(\boldsymbol X);\phi\right)\\
 h\left(\mu(\boldsymbol X)\right)&=&\alpha+\sum_{j=1}^{J}\sum_{k=1}^K \gamma_{jk}\beta_{jk} \rho_{k}(x_j).
  \end{array}
\end{equation} 

Note that the vector $M = \{\gamma_{jk}, j = 1, \dots, J, k = 1, \dots, K\}$ fully characterizes a model, as it defines which predictors $x_j$ are included in the model and after which transformation $\rho_{k}$. It is then sufficient to define priors on $M$ and on the related (read given $M$) coefficients $\alpha$ and $\beta_{jk}$ to complete the procedure.

Let us start by defining the prior for $M$,

\begin{equation}\label{eq:modelprior}
P(M)\propto\ \mathbb{I}\left(|M|\leq q\right)\prod_{j=1}^J\prod_{k=1}^K \mathbb{I}\left(\left[\sum_{k = 1}^K \gamma_{jk}\right]\leq d\right)a_k^{\gamma_{jk}},
\end{equation}

where $q, d \in \mathcal{N}$ and $0 < a_k < 1$, $k = 1, \dots, K$. Here $|M| = \sum_{j=1}^J \sum_{k = 1}^K \gamma_{jk}$ is the total number of terms included in the model, which can be bounded by $q$ to favor sparse models, and $a_k^{\gamma_{jk}}$ are prior penalties on the individual terms. More interestingly, here $\mathbb{I}(\sum_{k = 1}^K\gamma_{jk}\leq d)$ are (common for each predictor) prior indicators which restrict the number of terms per predictor to be simultaneously included into the model. %Basically, here $d$ controls the order of the fractional polynomials, with $d = 1$ only allowing for one polynomial term per predictor, i.e., the classical definition of a fractional polynomial model. At the same time, $d>1$ allows softer versions of prior penalties and thus more flexibility in modelling a fractional polynomial regression. Thus, $q$ and $d$ are defining prior constraints on the models.

%Just as described in \citet{hubin2020novel} and \citet{hubin2021flexible}, if $M$ and $M'$ are two vectors satisfying the constraints induced by $p$ and $q$ and differing in one component, say $\gamma_{j'k'}=1$ and $\gamma_{jk} = 0$, then
%\[
%\frac{P(M')}{P(M)}=a<1
%\]  
%showing that larger models are penalized more. This result easily generalizes to the comparison of more different models  and provides the basic intuition behind the chosen prior. 

As a fully Bayesian approach, BGNLM also requires priors on the parameters. Here we follow \citet{hubin2020novel} and use the common improper prior \citep{LiClyde2018,bayarri2012criteria} $\pi(\phi) = \phi^{-1}$ to the unknown dispersion parameter $\phi$, and simple Jeffreys priors ~\citep{jeffreys1946invariant,gelman2013bayesian} $|\mathcal{J}^M_n(\alpha,\boldsymbol{\beta})|^{\frac{1}{2}}$ on the regression parameters, where $\mathcal{J}^M_n(\alpha,\boldsymbol{\beta})$ is the observed information for the model $M$. 

Jeffreys prior is known to have attractive properties of being objective and scale invariant \citep{gelman2013bayesian}. Moreover, when using Jeffreys prior, the marginal likelihood of a model $P(Y| M)$ can be approximated accurately using the Laplace approximation. In the case of a Gaussian model, if we choose the aforementioned priors for the dispersion parameter and the coefficients, the Laplace approximation becomes exact \citep{raftery1997bayesian}. This results in a marginal likelihood of the simple form
\begin{equation} \label{MargLik_Jeffrey}
P(Y| M) \propto P(Y| M,\hat\theta)\ n^{\frac{|M|}{2}},
\end{equation}
where $\hat\theta$ refers to the maximum likelihood estimates of all parameters involved and $n$ is the sample size. On the log scale, this corresponds exactly to the BIC model selection criterion \citep{Schwarz} when using a uniform model prior.

%In the case of logistic regression, the marginal likelihood under Jeffreys prior is approximately \eqref{MargLik_Jeffrey} with an error of order $O(n^{-1})$ \citep{KadaneErr,claeskens_hjort_2008}. As noted by \citet{Barber2016}, Laplace approximations of the marginal likelihood yield very accurate results and can be trusted in Bayesian model selection problems.

\subsection{Consistency of model selection under our priors}

Assume $a_k = \exp(-s_k \log n)$, where $s_k$ is an arbitrary positive and finite scalar.  This is a special case of what is called BIC-type penalization of complexity in \citet{hubin2021flexible} that we shall also use in the experimental sections. Let $\hat{\theta}_i$ be the maximum likelihood estimate (MLE) of the parameters in model $M_i$. Define $p_i = |M_i| + \sum_{j=1}^J\sum_{k = 1}^K\gamma^i_{jk}{2s_k}$ as the sum of the standard BIC penalty \citet{claeskens_hjort_2008} and our prior penalty (excluding $\log n$) for $M_i$. Define $\text{PIC}_i$ to be a negative log posterior (up to a constant) of $M_i$ under Laplace approximations, i.e. as 
\begin{equation}\label{eq:PIC}
 \text{PIC}_i = -2l(\hat{\theta}_i|M_i) + p_i\log n,   
\end{equation}
where $l({\theta}_i|M_i) = \log P(Y|{\theta}_i,M_i)$ is the log-likelihood of the data given the MLE of parameters in model $M_i$. Then $P(M_i|Y)\propto \exp(-\tfrac{\text{PIC}_i}{2})$. Assuming the true model $M_0$ is in the set of candidate models and does not coincide with the null model, in the following proposition we show that as $n\rightarrow\infty$, the probability of selecting the true model $M_0$ among the candidate models goes to one.

\begin{thm}
Let $a_k = \exp(-s_k \log n)$ with $0<s_k<\infty$. Let $M_0$ be the unique parsimonious true model living on our model space $\mathcal{M}$ which has $1<|M_0|\leq q$ and $\sum_{k = 1}^K \gamma^0_{jk} \leq d, \forall j = 1,\dots,J$. Further, let $M_1,\dots,M_K \subseteq \mathcal{M}$ be the set of candidate models on $\mathcal{M}$ that satisfy constraints induced by $q$ and $d$. Then PIC criterion from equation \eqref{eq:PIC} is consistent in selecting $M_0$ among  $M_1,\dots,M_K$.
\end{thm}
\begin{proof}

Let $A_n$ be the event that the PIC selects the true model $M_0$, i.e.:
$A_n = \mathbb{I}\{\text{PIC}_0 < \text{PIC}_1, \dots, \text{PIC}_K\}.$
We want to show that $\text{P}(A_n = 1)\rightarrow 1$ as $n\rightarrow\infty$.
By the law of large numbers and the continuous mapping theorem, we have
$\underset{n\rightarrow\infty}{\text{plim}}  \frac{1}{n}l(\hat{\theta}_0|M_0) =  \mathbb{E}_0[l(\theta_0|M_0)],$
where $\text{plim}$ is the limit in probability operator and $\mathbb{E}_0$ denotes the expectation under the true model $M_0$.
Therefore, we have
$$\underset{n\rightarrow\infty}{\text{plim}}\text{PIC}_0 = \underset{n\rightarrow\infty}{\text{plim}} -2l(\hat{\theta}_0|M_0) + p_0\log n = -2n\mathbb{E}_0[l(\theta_0|M_0)] + p_0\log n.$$
Similarly,
$\underset{n\rightarrow\infty}{\text{plim}}\text{PIC}_i = -2n\mathbb{E}_i[l(\theta_i|M_i)] + p_i\log n,$
where $\mathbb{E}_i$ denotes the expectation under model $M_i$. Taking the limiting difference between $\text{PIC}_0$ and $\text{PIC}_i$, we have
$$\underset{n\rightarrow\infty}{\text{plim}} \text{PIC}_0 - \text{PIC}_i = -2n\Delta_i + C_i\log n = -\infty,\forall i: M_i \in \{M_1,\dots,M_K\} \setminus M_0\subseteq \mathcal{M}$$
where $\Delta_i = \mathbb{E}_0[l(\theta_0|M_0)] - \mathbb{E}_i[l(\theta_i|M_i)]$ and $C_i = p_0 - p_i$. In other words, we show that
$$\underset{n\rightarrow\infty}{\text{plim}} (A_n) = \underset{n\rightarrow\infty}{\text{plim}} \mathbb{I}(\text{PIC}_0 < \text{PIC}_i, i\neq 0) = \underset{n\rightarrow\infty}{\text{plim}}\mathbb{I}(\Delta_i > \frac{1}{2n}C_i\log n, i\neq 0) = 1.$$ Here, $\Delta_i \geq 0$ and we have two cases to check: 1. if $\Delta_i=0$ then $M_0$ is a nested model of $M_i$ and hence $C_i<0$ by the uniqueness and parsimony of the true model. 2. For $\Delta_i > 0$ it is sufficient that $\lim_{n\rightarrow \infty} \frac{1}{2n}(C_i)\log n = 0,\forall C_i: |C_i|<\infty$, and $C_i$ is always finite by the construction of our priors and the fact that $J<\infty$ and $K<\infty$. Thus, we have shown that the PIC criterion is consistent in selecting the true model as the sample size increases.
\end{proof}

\subsection{Bayesian fractional polynomials models as Bayesian Generalized Nonlinear Models}\label{specialCase}

In order to recover our fractional polynomials models, we just need to specify the appropriate set of transformations $\mathcal{D}$ and parameters in the prior on $M$ defined in equation \eqref{eq:modelprior}. The parameters of the latter, in particular, control both the order of the fractional polynomials and the model selection mechanism.

\paragraph{Set of transformations.} As per the definition of fractional polynomials, the following transformations of each predictor are allowed: the identity, $\mathbf{F}_0=\{x\}$; 7 simple functions $\mathbf{F}_1= \{x^{-2}, x^{-1}, x^{-0.5}$, $\log x, x^{0.5}, x^{2}, x^{3}\}$; and 8 functions specifying repeated powers $\mathbf{F}_2 = \{x^{-2}\log x$, $x^{-1}\log x, x^{-0.5}\log x$, $\log x\log x, x^{0.5}\log x$, $x\log x$, $x^{2}\log x, x^{3}\log x\}$. If we want to fit fractional polynomials of order 1, then $\mathcal{D} = \{\mathbf{F}_0 \cup \mathbf{F}_1\}$, while for fractional polynomials of order 2, $\mathcal{D} = \{\mathbf{F}_0 \cup \mathbf{F}_1 \cup \mathbf{F}_2\}$. In this framework, it is straightforward to increase the order of the fractional polynomials by adding further interaction terms, but according to \citet[][Ch. 5.9]{RoystonSauerbrei2008} they are not used in practice, and we do not consider them here. 

\paragraph{Order of the fractional polynomials.} The order of the fractional polynomials is basically controlled by the value of the parameter $d$ in the prior \eqref{eq:modelprior}, as $d$ is the maximum number of transformations that are allowed for each explanatory variable. So, $d = 1$, by only allowing for one polynomial term per variable, leads to a FP(1), while $d > 1$ allows more flexibility in modeling a fractional polynomial regression, up to the desired order. Note that in the case of order, $d$ larger than 1, additional modifications of the prior on $M$ are necessary if one wants to exclude combinations in $\mathbf{F}_2$ without the corresponding term in $\mathbf{F}_1$. This can be easily done by forcing the model priors including such transformations to be 0.

\paragraph{Model selection.} In fractional polynomials models there are two sources of complexity to take into account when performing a model selection procedure: the number of regression parameters and the degree of the transformations. Following a paradigm of parsimony, one would ideally consider including variables only if they are related to the response and transforming the variables only if needed. In the framework of BGNLM, there are two separate (sets of) parameters to control the two sources: the parameter $q$ sets an upper bound to the number of active components of the models (i.e. the predictors and their transformations), while the $a_k$, $k=1, \dots, K$, set a cost to include individual predictors and can therefore be used to penalize harder more complex transformations. The penalty to add a linear term (the transformation belonging to $\mathbf{F}_0$) is set to be much lower than that for adding a transformation (those in $\mathbf{F}_1$ and $\mathbf{F}_2$). Fractional polynomials of order 2 are naturally penalized more as they require 2 terms, i.e., the penalty on the log scale has the form $\log a_{k'} + \log a_{k''}$.

\subsection{Model fitting via the Genetically Modified Mode Jumping MCMC}\label{sec:inference}

In fractional polynomials regression models with the mean parameter linked to the data through Equation \eqref{eq:LRM}, the model space can become prohibitively large even with only moderate values of the number of candidate predictors $J$. The strong correlation among predictors (especially that between different transformations of the same variable), moreover, can lead to many local minima in the posterior distribution, in which a standard fitting algorithm such as the classical MCMC may get stuck with a higher probability. To address these issues, we implement the Genetically Modified Mode Jumping MCMC algorithm proposed by \cite{hubin2020novel} to fit Bayesian fractional polynomials models. 

The key idea behind GMJMCMC is to iteratively apply a Mode Jumping MCMC algorithm \citep{hubin2016efficient} to smaller sets of model components of size $s: q \leq s \ll 16J$. This reduces the number of models in the model space to $\sum_{k=1}^q \genfrac(){0pt}{2}{16J}{k}$. A sequence of so-called populations $\mathcal{S}_1, \mathcal{S}_2, \dots, \mathcal{S}_{T_{max}}$ is generated. Each population $\mathcal{S}_t$ is a set of $s$ transformations and forms a separate search space for exploration through Mode Jumping MCMC iterations. The populations dynamically evolve allowing GMJMCMC to explore different parts of the total model space.

The generation of the new population $\mathcal{S}_{t+1}$ given $\mathcal{S}_t$ works as follows: some components with low posterior probability from the current population are removed, and then replaced by new components generated through mutation, multiplication, modification, or projection operators. The probabilities for each operator are defined as $P_{in}$, $P_{mu}$, $P_{mo}$, and $P_{pr}$, respectively, and must add up to 1. Since fractional polynomials are a specific case of BGNLM with only modification transformations allowed, in this context the algorithm is simplified by setting $P_{mu}=0$ and $P_{pr} = 0$. The algorithm is summarized below.

\begin{algorithm}[h]
\small
\caption{GMJMCMC}\label{gMJMCMCalg}
\begin{algorithmic}[1]
\State Initialize $\mathcal{S}_{0}$
\State Run the MJMCMC algorithm within the search space $\mathcal{S}_0$ for $N_{init}$ iterations and use results to initialize $\mathcal{S}_1$.
\For{$t=1,...,T-1$} 
\State Run the MJMCMC algorithm within the search space $\mathcal{S}_t$ for $N_{expl}$ iterations.\State Generate a new population $\mathcal{S}_{t+1}$
\EndFor
\State Run the MJMCMC algorithm within the search space $\mathcal{S}_{T}$ for $N_{final}$ iterations.
\end{algorithmic}
\end{algorithm}

For a complete description of the Mode Jumping MCMC, including its theoretical properties, we refer the reader to \citet{hubin2016efficient}. For further details on the GMJMCMC, see \citet{hubin2020novel}, while results on its asymptotic exploration of the space of nonlinear models are available in \citet{hubin2021flexible}.

\subsection{Using the output of GMJMCMC to compute the marginals of interest}

The posterior probability of a model $M$ given the observed data $Y$ can be expressed as the product of the prior probability of the model $P(M)$ and the marginal likelihood of the data given the model $P(Y|M)$ divided by the sum of the same expression over all possible models in the model space $\mathcal{M}$, which is infeasible to explore. To approximate this, the GMJMCMC algorithm searches for a set of good models $\Omega\subseteq\mathcal{M}$, and the resulting approximation for the posterior probability of a model $M$ given the data $Y$ is denoted as $$\widehat{P}(M|Y) = \frac{P(Y|M)P(M)}{\sum_{M'\in\Omega}P(Y|M')P(M')}.$$ The marginal inclusion probabilities for a specific effect $\gamma_{jk}$, denoted as $\widehat{P}(\gamma_{jk}=1|Y)$, can then be calculated as the sum of the approximated posterior probabilities over all models in $\Omega$ that include this effect, i.e. $$\widehat{P}(\gamma_{jk}=1|Y) = \sum_{M\in \Omega: \gamma_{jk} = 1} \widehat{P}(M | Y).$$ Further, marginal posterior of any other quantity of interest $\Delta$ can be approximated as $$\widehat{P}(\Delta|Y) =  \sum_{M\in\Omega}P(\Delta|Y,M)\widehat{P}(M|Y).$$ This allows us to make predictions based on the output of GMJMCMC. 

\subsection{Extensions of the model}

The description of the Bayesian fractional polynomials models seen so far only covers the GLM setting (see formula \eqref{eq:LRM}), but our approach can be easily extended to many other cases. Due to their particular practical relevance, here we cover the cases of generalized linear mixed models, the Cox regression model, and the model with interactions.

\subsubsection{Latent Gaussian models}

It is straightforward to extend our approach for generalized linear mixed models by incorporating both polynomial terms and latent Gaussian variables. These variables can be used to model correlations between observations in space and time, as well as over-dispersion. Basically, we just need to substitute the function $h$ in Equation \eqref{eq:LRM} with

\begin{equation}\label{BLRModel3}
\mathsf{h}\left(\mu\left(\boldsymbol{X}\right)\right) = \alpha + \sum_{j=1}^{J}\sum_{k = 1}^K \gamma_{jk}\beta_{jk} \rho_{k}(x_j)+\sum_{r=1}^{R}\gamma_{JK+r}\delta_r,
\end{equation}

where $\delta_r \sim N\left(\boldsymbol{0}, \boldsymbol\Sigma_r\right)$ are latent Gaussian variables with covariance matrices $\boldsymbol{\Sigma}_r$. These variables allow us to model different correlation structures between individual observations. The matrices typically depend only on a few parameters $\psi_r$, so that in practice we have $\boldsymbol{\Sigma}_r=\boldsymbol{\Sigma}_r(\boldsymbol\psi_r)$.

The model priors of Equation~\eqref{eq:modelprior} needs to be generalised to handle inclusion indicators $\gamma_{JK+r}, r = 1,\dots,R$ of the latent Gaussian variables and becomes:
\begin{equation*}
P(M)\propto\ \mathbb{I}\left(|M|\leq q\right)\prod_{j=1}^J\prod_{k=1}^K \mathbb{I}\left(\left[\sum_{k = 1}^K \gamma_{jk}\right]\leq d\right)a_k^{\gamma_{jk}}\prod_{r=1}^R \Tilde{a}_r^{\gamma_{JK+r}},
\end{equation*}
where $\Tilde{a}_r$ are prior inclusion penalties for the corresponding latent Gaussian variables. 

The parameter priors are adjusted as follows:
\begin{align}
\boldsymbol{\beta}|\boldsymbol{\gamma} \sim & N_{p_{\boldsymbol{\gamma}}}(\boldsymbol{0},I_{p_\gamma}e^{-\psi_{\beta_\gamma}}),\\
\boldsymbol{\psi}_k\sim&\pi_k(\boldsymbol{\psi}_k).\label{latentprior}
\end{align}

We can choose any type of hyperparameters of priors that are compatible with the integrated nested Laplace approximations (INLA) \citep{rue2009eINLA}. This allows us to efficiently compute the marginal likelihoods of individual models using the INLA approach \citep{HubinStorvikINLA}. Additionally, any other extensions with computable marginal likelihoods are possible within our framework, as the availability of the marginal likelihood is sufficient to run our inference algorithm described in Section~\ref{sec:inference}.

\subsubsection{Cox regression model}

Our approach can also be used to analyze time-to-event data, for example by using the Cox regression model. Here the adaptation of the formula in Equation \eqref{eq:LRM} is not straightforward, as the Cox regression model works with hazards and not densities,
$$
\lambda(y; \mu(\boldsymbol X)) = \lambda_0(y) \exp\{\mu(\boldsymbol X)\},
$$
where $\lambda_{0}(y)$ is the so-called baseline hazard function, i.e. that function that models the part of the hazard that does not depend on the predictors (including the intercept, we will not have a parameter $\alpha$ here). An additional complication with the time-to-event analysis is the presence of censored observations, i.e. those statistical units for which we only have part of the information (typically, that the event of interest happens after the last observed time). For our model fitting procedure, however, we just need to know the (partial) likelihood
\begin{equation}\label{CoxL}
L(\mu(\boldsymbol X)) = \prod_{i = 1}^n \frac {\mu(X_i)}{\sum_{r \in R(y_i)} \mu(X_r)},
\end{equation}
where $R(y_i)$ includes all the observations at risk at the time $y_i$, and consider
$$
\log(\mu(\boldsymbol X)) = \sum_{j=1}^{J}\sum_{k=1}^K \gamma_{jk}\beta_{jk} \rho_{k}(x_j).
$$
In the case of a partial likelihood in the form of Equation \eqref{CoxL} there is a handy approximation of the marginal likelihood provided by \cite{RafteryAl2005}. Here we used their results. The priors on the model and the parameter, instead, are the same as those of Section \ref{specialCase}.

\subsubsection{Fractional polynomials with flexible intercations}\label{interactions}

As a specific case of BGNLM, our BGNLM\_FP can be easily generalized to handle interactions up to a given order $I$ between the polynomial terms of different predictors, which results in the following generalization of our model

\begin{align*}
    h\left(\mu(\boldsymbol X)\right)=&\alpha+\sum_{j=1}^{J}\sum_{k=1}^K \gamma_{jk}\beta_{jk} \rho_{k}(x_j)+\sum_{j=1}^{J}\sum_{k=1}^K\sum_{j^{(1)}=1}^{J}\sum_{k^{(1)}=1}^K \gamma^{(1)}_{jkj^{(1)}k^{(1)}}\beta^{(1)}_{jkj^{(1)}k^{(1)}} \left[\rho_{k}(x_j)\times \rho_{k^{(1)}}(x_{j^{(1)}})\right] + \\
   & \sum_{j=1}^{J}\sum_{k=1}^K\sum_{j^{(1)}=1}^{J}\sum_{k^{(1)}=1}^K\sum_{j^{(2)}=1}^{J}\sum_{k^{(2)}=1}^K \gamma^{(2)}_{jkj^{(1)}k^{(1)}j^{(2)}k^{(2)}}\beta^{(2)}_{jkj^{(1)}k^{(1)}j^{(2)}k^{(2)}} \left[\rho_{k}(x_j)\times \rho_{k^{(1)}}(x_{j^{(1)}}) \times \rho_{k^{(2)}}(x_{j^{(2)}}) \right]\\ 
   &+\dots +\\
   & \sum_{j=1}^{J}\sum_{k=1}^K\sum_{j^{(1)}=1}^{J}\sum_{k^{(1)}=1}^K\sum_{j^{(2)}=1}^{J}\sum_{k^{(2)}=1}^K \dots \sum_{j^{(I)}=1}^{J}\sum_{k^{(I)}=1}^K \gamma^{(I)}_{jkj^{(1)}k^{(1)}j^{(2)}k^{(2)}\dots j^{(I)}k^{(I)}}\beta^{(I)}_{jkj^{(1)}k^{(1)}j^{(2)}k^{(2)}\dots j^{(I)}k^{(I)}}\\& \left[\rho_{k}(x_j)\times \rho_{k^{(1)}}(x_{j^{(1)}}) \times \rho_{k^{(2)}}(x_{j^{(2)}}) \times \dots \times \rho_{k^{(I)}}(x_{j^{(I)}}) \right] .
\end{align*}

Now, an extended vector $M = \{\gamma_{jk},\gamma^{(1)}_{jkj^{(1)}k^{(1)}},\gamma^{(2)}_{jkj^{(1)}k^{(1)}j^{(2)}k^{(2)}},...\gamma^{(I)}_{jkj^{(1)}k^{(1)}j^{(2)}k^{(2)}...j^{(I)}k^{(I)}}, j = 1, \dots, J, k = 1, c, K,j^{(i)} = 1, \dots, J, k^{(i)}= 1, \dots, K, i = 1, \dots I\}$ fully characterize a model with the order of interactions up to $I$. It then defines which predictors $X_j$, which transformation $\rho_{k}$ and which interactions between them are included in the model. Finally, we generalise the priors from Equation \eqref{eq:modelprior} by means of setting $d=\infty$ and defining $a_{kk^{(1)}}, a_{kk^{(1)}k^{(2)}}, \dots a_{kk^{(1)}k^{(2)}\dots k^{(I)}}$ as follows: 
\begin{align*}
    a_{kk^{(1)}} & = a_{k}\times a_{k^{(1)}}\\
    a_{kk^{(1)}k^{(2)}} & = a_{k}\times a_{k^{(1)}}\times a_{k^{(2)}}\\
    &\dots\\
    a_{kk^{(1)}k^{(2)}\dots k^{(I)}} &= a_{k}\times a_{k^{(1)}}\times a_{k^{(2)}} \times \dots \times a_{k^{(I)}}.
\end{align*}
The parameter priors here remain the same as those defined in Section \ref{specialCase}. And the inference is enabled by putting a non-zero value to the tuning parameter $P_{mu}>0$ in the GMJMCMC algorithm.

\section{Simulation studies}\label{sec:ART}

%\subsection{Settings}

%To evaluate numerically the consistency of our novel algorithm and contrast its performances with those of the current fractional polynomials models' implementations, we take advantage of the ART study \citep{RoystonSauerbrei2008}, an existing simulation design essentially created to assess fractional polynomials models. Based on a large breast cancer data set \citep{SchmoorAl1996}, the ART study provides a realistic framework when it concerns the distribution of the predictors and their correlation structure. See \citet[][Tables 10.1 and 10.3]{RoystonSauerbrei2008} for the details.

\subsection{Aims}
The primary goal of the simulation study is to evaluate numerically the consistency of our novel algorithm and contrast its performances with those of the current fractional polynomials models' implementations. In particular, we want to assess the ability of recovering the true data-generating process when increasing the signal-to-noise ratio, or, at least, of selecting the relevant variables.

\subsection{Data-generating mechanism}
We take advantage of the ART study \citep{RoystonSauerbrei2008}, an existing simulation design essentially created to assess fractional polynomials models. Based on a large breast cancer data set \citep{SchmoorAl1996}, the ART study provides a realistic framework when it concerns the distribution of the predictors and their correlation structure. See \citet[][Tables 10.1 and 10.3]{RoystonSauerbrei2008} for the details. More specifically, the ART study consists of 6 continuous ($x_1$, $x_3$, $x_5$, $x_6$, $x_7$, $x_{10}$) and 4 categorical explanatory variables, in turn, divided into an ordered three-level ($x_4$), an unordered three-level ($x_9$) and two binary ($x_2$ and $x_8$) variables. For an initial analysis of the data, we refer to \cite[][Chapter 10]{RoystonSauerbrei2008}. The response is computed through the model
$$
y = x_1^{0.5} + x_1 + x_3 + x_{4a} + x_5^{-0.2} + \log\{x_6 + 1\} + x_8 + x_{10} + \epsilon,
$$
where $x_{4a}$ denotes the second level of $x_4$ (the first being used as a baseline) and $\epsilon\sim N(0;1)$. The instances used in the original simulation study are available at \url{http://biom131.imbi.uni-freiburg.de/biom/Royston-Sauerbrei-book/Multivariable_Model-building/downloads/datasets/ART.zip} and directly used here. In total, the data set has $n=250$ observations. 

While the original study is interesting to evaluate the FP approach in a likely situation, it does not allow us to fully investigate the properties of our algorithm. For example, model misspecifications like $x_5^{-0.2}$ prevent us to evaluate how often the algorithm selects the true model. 

In order to study the properties of our algorithm, we propose a modification of the existing study. We change the original model by including a FP(-1) effect for $x_5$ (instead of $x_5^{-0.2}$, such that the true model belongs to the set of the possible models) and by modifying the effect of $x_3$ from linear to a FP2(-0.5, -0.5), to make the search for the true model more challenging. The new response generating mechanism now follows
\begin{equation}\label{trueGF}
y = x_1^{0.5} + x_1 + x_3^{-0.5} + x_3^{-0.5} * \log(x_3+\varepsilon) + x_{4a} + x_5^{-1} + \log(x_6 + \varepsilon) + x_8 + x_{10} + \epsilon,
\end{equation}
where again $x_{4a}$ denotes the second level of $x_4$ and $\varepsilon = 0.00001$ is a small positive real number used to avoid problems with the support of a logarithm function. As in the original formulation \citep[][Ch. 10]{RoystonSauerbrei2008}, the true regression coefficients are all set equal to 1. Finally, $\epsilon\sim N(0,\sigma^2)$ and we run 16 different scenarios with $\sigma^2 \in \{100, 50, 25, 10, 1, 0.1, 0.01, 10^{-3}, 10^{-4}, 10^{-5}, 10^{-6}, 10^{-7}, 10^{-8}, 10^{-9}, 10^{-10}\}$ allowing to quantify the consistency of model selection for the compared approaches when increasing signal-to-noise ratios (that, under the Gaussian distribution, is mathematically equivalent to increasing the sample size).

\subsection{Estimands}
The quantity of interest in this simulation study is the model's ability to recover the significant relationships between predictors and responses. We distinguish between functional (strict) and predictor (soft) levels. In the former, we are interested in how often the model selects the true functional form (and how often a wrong one); in the latter, in how often the model selects a relevant variable (no matter in which form) and how often it incorrectly selects an irrelevant variable. The variable level is useful because the difference between functional forms may be quite small (consider, for example, the logarithmic and the square root transformation) and therefore not impact too much the result.

For each method, we consider the transformations and the variables from the model that the respective method declared to be (with respect to the criterion used in the method) optimal for a data set. For the Bayesian approaches, a variable is classified as detected if the (estimated) marginal inclusion probability is larger than 0.5. This corresponds to the median probability model of \citet{BarbieriBerger2004}.

%\subsection{Competitors and software implementation}
\subsection{Methods}

Our novel model, hereafter BGNLM\_FP, was fitted by GMJMCMC algorithm using the \texttt{EMJMCMC} package available at \textit{\texttt{http://aliaksah.github.io/EMJMCMC2016/}}. The simulations for each $\sigma^2$ were run on $32$ parallel threads $100$ times. Each thread was run for $20,000$ iterations with a mutation rate of $250$ and the last mutation at iteration $15,000$. The population size of the GMJMCMC algorithm was set to $20$. For the detection of the functional forms, the median probability rule \citep{BarbieriBerger2004} was used. For all simulation scenarios, we specified the following values of  the hyper-parameters of the model priors:  $q=20$ and $d=16$. Further, $a_k$ was chosen to be $a_k = \exp(-\log n)$ for $k: \rho_k \in \mathbf{F}_0$, $a_k = \exp(- (1+\log 2)\log n)$ for $k: \rho_k \in \mathbf{F}_1$, and $a_k = \exp(- (1+\log 4)\log n)$ for $k: \rho_k \in \mathbf{F}_2$. No additional fine-tuning was performed to specify our tuning and hyper-parameters.

For comparison, the frequentist version of multivariate fractional polynomials (MPF) was fitted using the R package \texttt{mfp} \citep{HeinzeAl2022}. We allowed for fractional polynomials of maximal order $2$, and used a significance level $\alpha = 0.05$. 

The current Bayesian version of fractional polynomials by \cite{SabanesboveHeld2011} was fitted using the R package \texttt{bfp} \citep{SabanesboveaL2022}, with ``flat'' (BFP\_F), ``sparse'' (BFP\_S), and ``dependent'' (BFP\_D) priors. We used the default value of 4 as the hyperparameter for hyper-g prior and the option ``sampling'' to explore the posterior model space. In this case as well, the median probability rule used to detect the functional form.

\subsection{Performance metrics}
To estimate our estimands, and consequently evaluate the performance of the different algorithms, we compute:
\begin{itemize}
 \item True Positive Rate (TPR), measuring the ability to detect the truly relevant predictors;
 \item False Discovery Rate (FDR), measuring the incorrect selections %performed by the model.
 of predictors that are not relevant.
\end{itemize}

Both TPR and FDR are computed at the functional and predictor levels. In the former level, a find is considered a ``true positive'' only if the model includes the true variable with the correct transformation. At the latter level, instead, we consider it sufficient to select a relevant predictor (so even if the model includes the variable with an incorrect transformation). Similarly, for the FDR, a ``false positive'' is any functional form that is not included in the model generative function (functional level) or a variable that is not included at all (variable level).

%Both TPR and FDR are computed at the functional (strict) and predictor (soft) levels. The former level is the most difficult to handle, as a find is considered a ``true positive'' only if the model includes the true variable with the correct transformation. At the latter level, instead, we consider it sufficient to select a relevant predictor (so even if the model includes the variable with an incorrect transformation). Similarly, for the FDR, a ``false positive'' is any functional form that is not included in the model generative function (functional level) or a variable that is not included at all (variable level). The variable level is useful because the difference between functional forms may be quite small (consider, for example, the logarithmic and the square root transformation) and therefore not impact too much the result.

%In order to give an insight into the variability, for each scenario and for each of the compared methods, we used a simple bootstrapping approach based on 100 resamples with replacement of 20 simulations to obtain medians and $95\%$ confidence intervals of the evaluation metrics.

\subsection{Results}

\begin{figure}[!ht]
  \centering
  \includegraphics[trim={1cm 0 1cm 0},clip,scale=0.5]{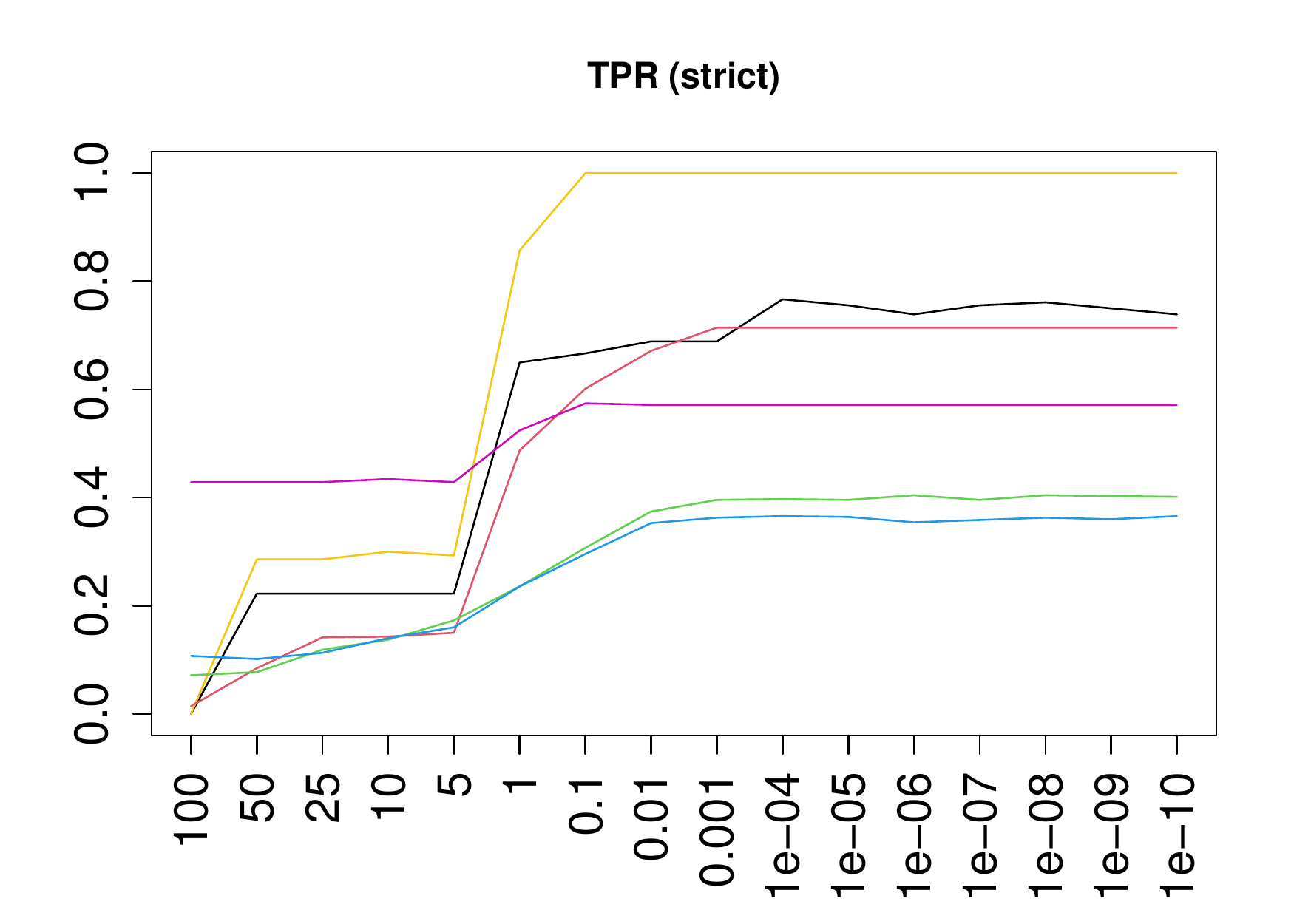}
  \includegraphics[trim={1cm 0 1cm 0},clip,scale=0.5]{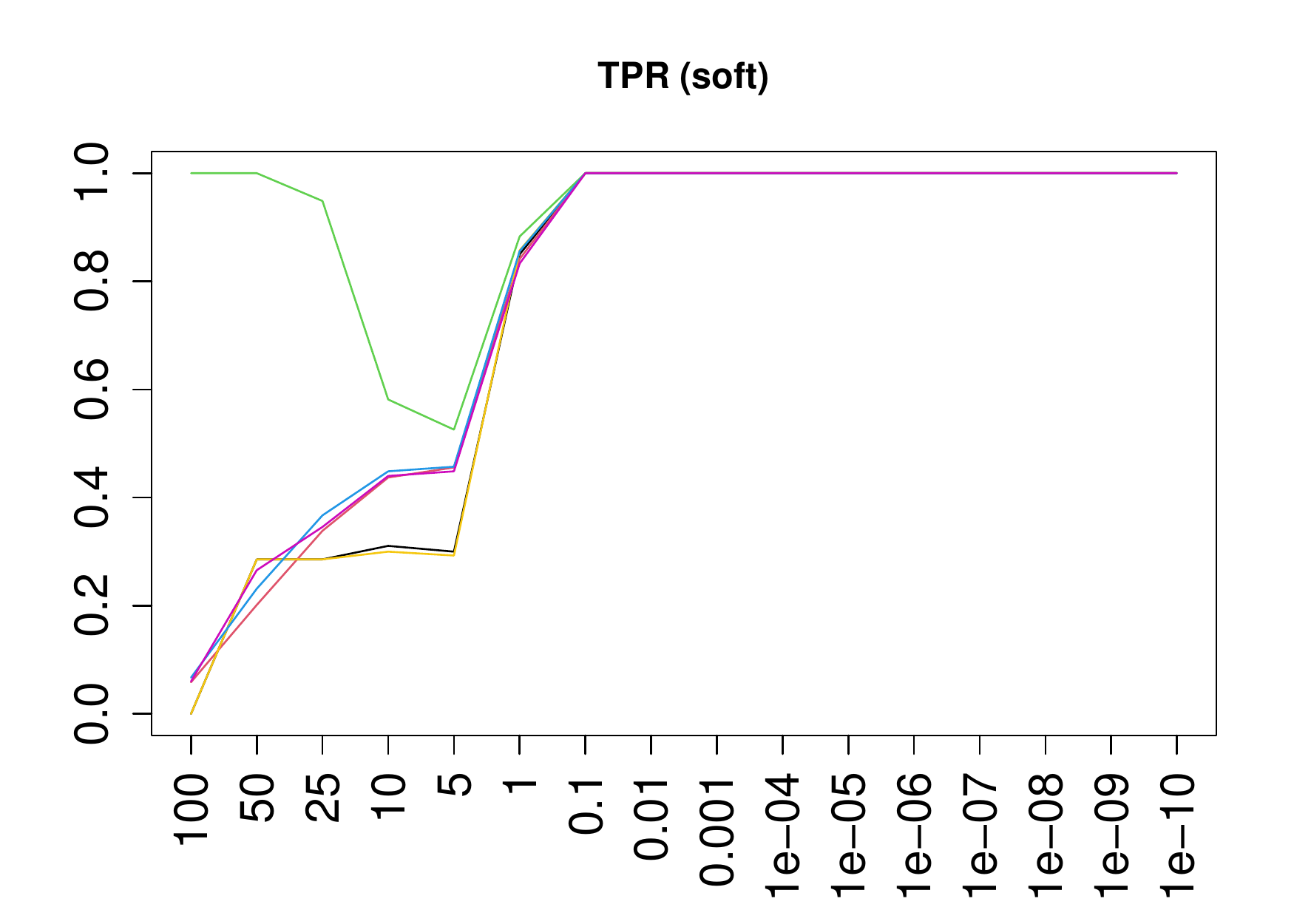}
  \includegraphics[trim={1cm 0 1cm 0},clip,scale=0.5]{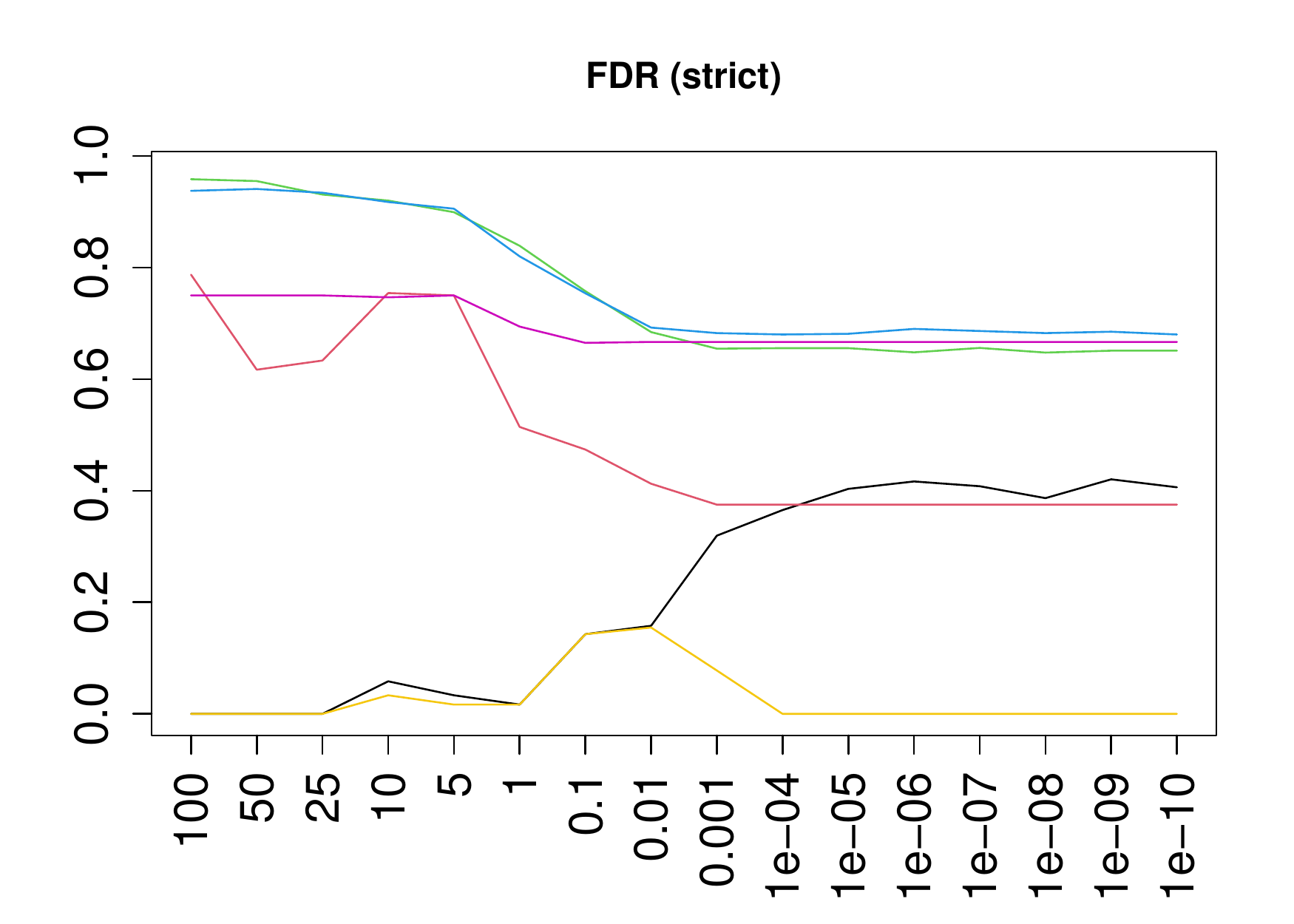}
  \includegraphics[trim={1cm 0 1cm 0},clip,scale=0.5]{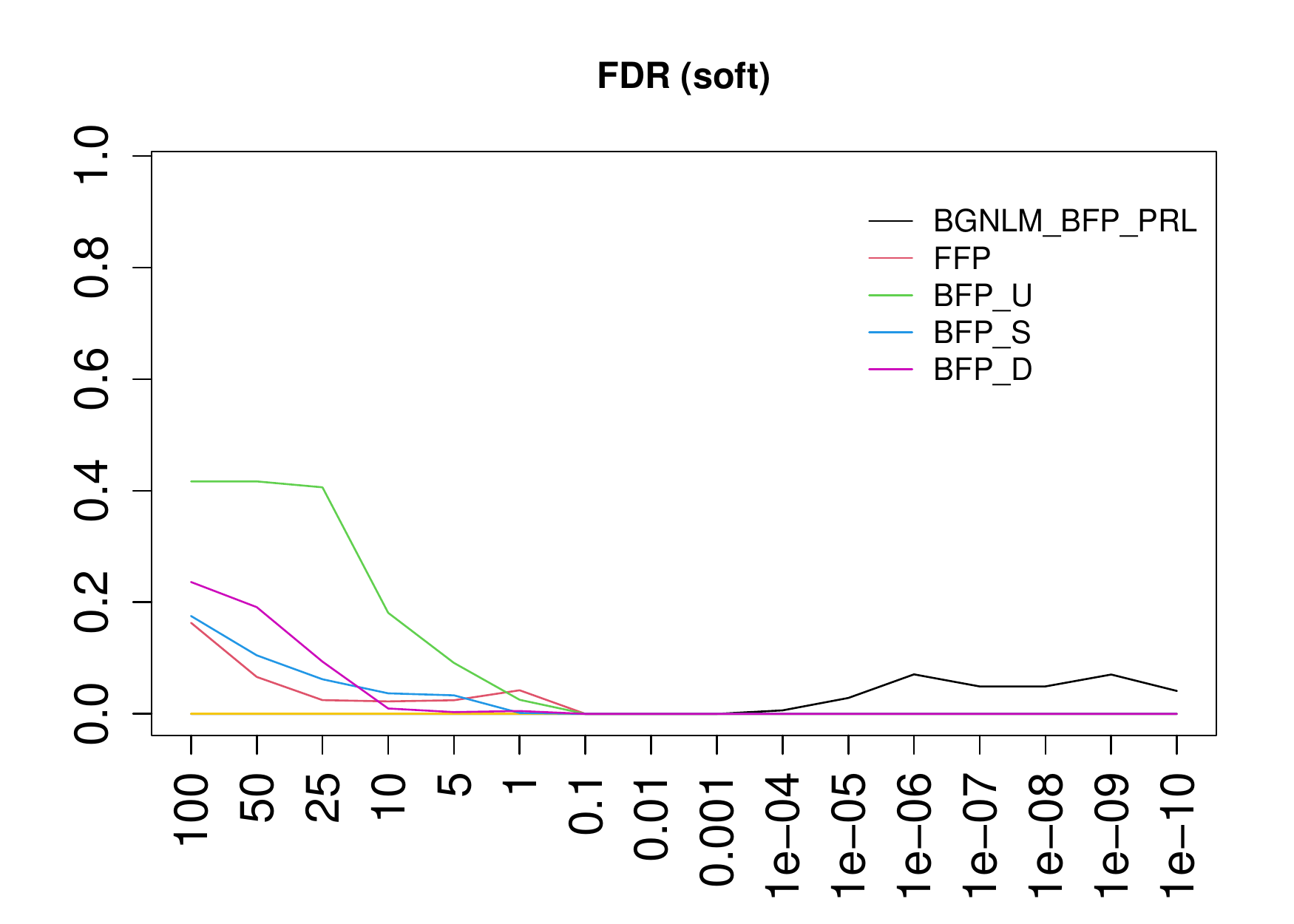}     \caption{\label{fig:pmeans} TPR (top row) and FDR (bottom row), at the functional (left panel) and at the variable (right) level, for: BGNLM\_FP (black), MFP (red), BFP with ``flat'' (BFP\_F, green), ``sparse'' (BFP\_S, blue), and ``dependent'' (BFP\_D, purple) priors. The yellow line shows BGNLM\_FP if the true model was forced into its search (BGNLM\_FP\_IDEAL).}
\end{figure}

In Figure~\ref{fig:pmeans}, top row, we see that the TPR grows, both at the functional (left panel) and at the variable level (right panel), for all approaches when the signal-to-noise ratio increases. At the functional level, our approach (BGNLM\_FP) and MFP are uniformly outperforming the current Bayesian approaches when there is enough signal (from $\sigma^2 \approx 1$ on), while BFP with a data-dependent prior (BFP\_D) works the best for low signal-to-noise data. Note, moreover, that BGNLM\_FP has a better performance than MFP in almost all scenarios, with the single exception of $\sigma = 0.001$. At the variable level, all approaches have a comparable TPR, with the notable exception of BFP\_F on the low signal-to-noise ratios. In fact, in these cases, BFP with a flat prior selects all variables with a linear effect, which raises some doubts about the choice of this prior in this context. Obviously, as a consequence, the FDR for BFP\_F at the variable level (Figure~\ref{fig:pmeans}, bottom right panel) is worse than all the others. 

At both the functional and variable level, we see the FDR decreasing with larger signal-to-noise for both MFP and BFP (bottom line of Figure~\ref{fig:pmeans}, left and right panels, respectively). While noticeably smaller than that of the competitors in almost all cases, a bit surprisingly the FDR of BGNLM\_FP becomes larger for a stronger signal. This counter-intuitive behavior is most probably related to the strong correlation between the true functional forms and other FP transformations. Even GMJMCMC seems to be stuck in some local extrema for lower noise levels, which under the same number of iterations did not allow it to reach the close neighborhood of the true model. When we force the search path to also include the true model (BGNLM\_FP\_IDEAL), in contrast, the median probability model correctly identifies it as the best option (see the yellow line in the left panels of Figure~\ref{fig:pmeans}): for noise levels smaller or equal to $\sigma^2 = 0.0001$, the TPR is 1 and the FDR is 0, showing in the practice of consistency of model selection under our family of priors. Also, the problem with FPR is almost negligible at the variable level (right panels of Figure~\ref{fig:pmeans}).

Figure \ref{fig:mlik} displays this BGNLM\_FP's behavior through the maximum value of the posterior. The GMJMCMC algorithm, whose search for the best model reaches the largest possible value (that of the true model, red in the figure) for lower signal-to-noise ratios, at certain points settles for a ``good enough'' value and does not reach anymore the largest one given the same number of iterations. This can be explained by that with a larger signal the correlation between the response and almost right fractional polynomials increases leading to stronger local extrema from which it is becoming harder to escape.

\begin{figure}[!ht]
    \centering
    \includegraphics[trim={1cm 0 1cm 0},clip,scale=0.5]{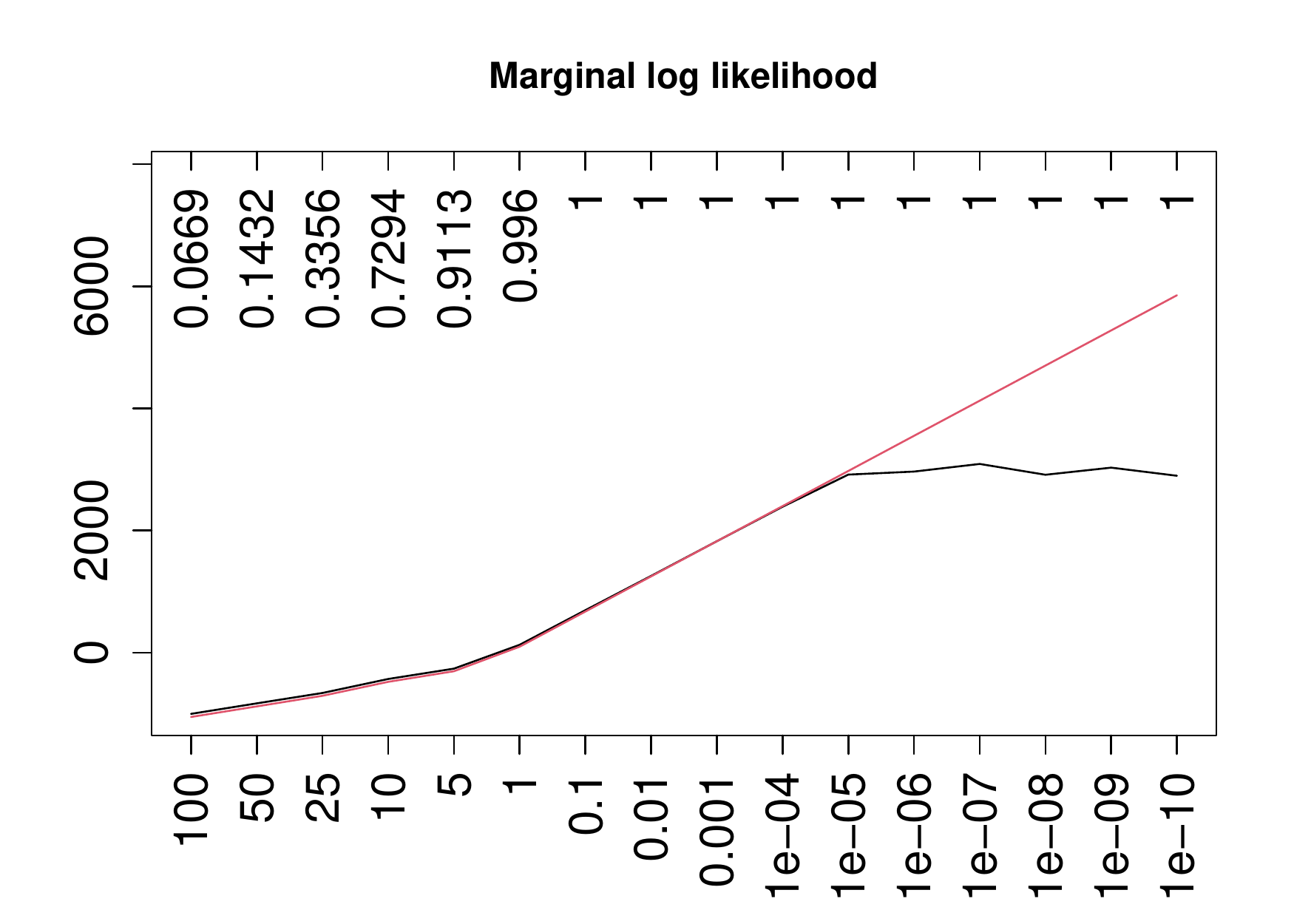}
    \caption{\label{fig:mlik} Best log marginal posteriors found with GMJMCMC (black) and those of a data generative model (red). Upper axis report $R^2$ of the true model.}
\end{figure}

To better appreciate the performance of all approaches on a single dimension, we also report in Figure~\ref{fig:x3} the TPR for the variable $X_3$. The left panel shows the TPR at the variable level: the task of identifying $X_3$ as relevant seems pretty easy, as all reasonable (remember that BFP\_F includes all variables) approaches start to include it when $\sigma^2$ is between $0.1$ and $1$. Even giving an unfair advantage in the case of BGNLM\_FP\_IDEAL does not change much, as the same process happens at almost the same signal-to-ratio level.

\begin{figure}[!ht]
    \centering
   \includegraphics[trim={1cm 0 1cm 0},clip,scale=0.5]{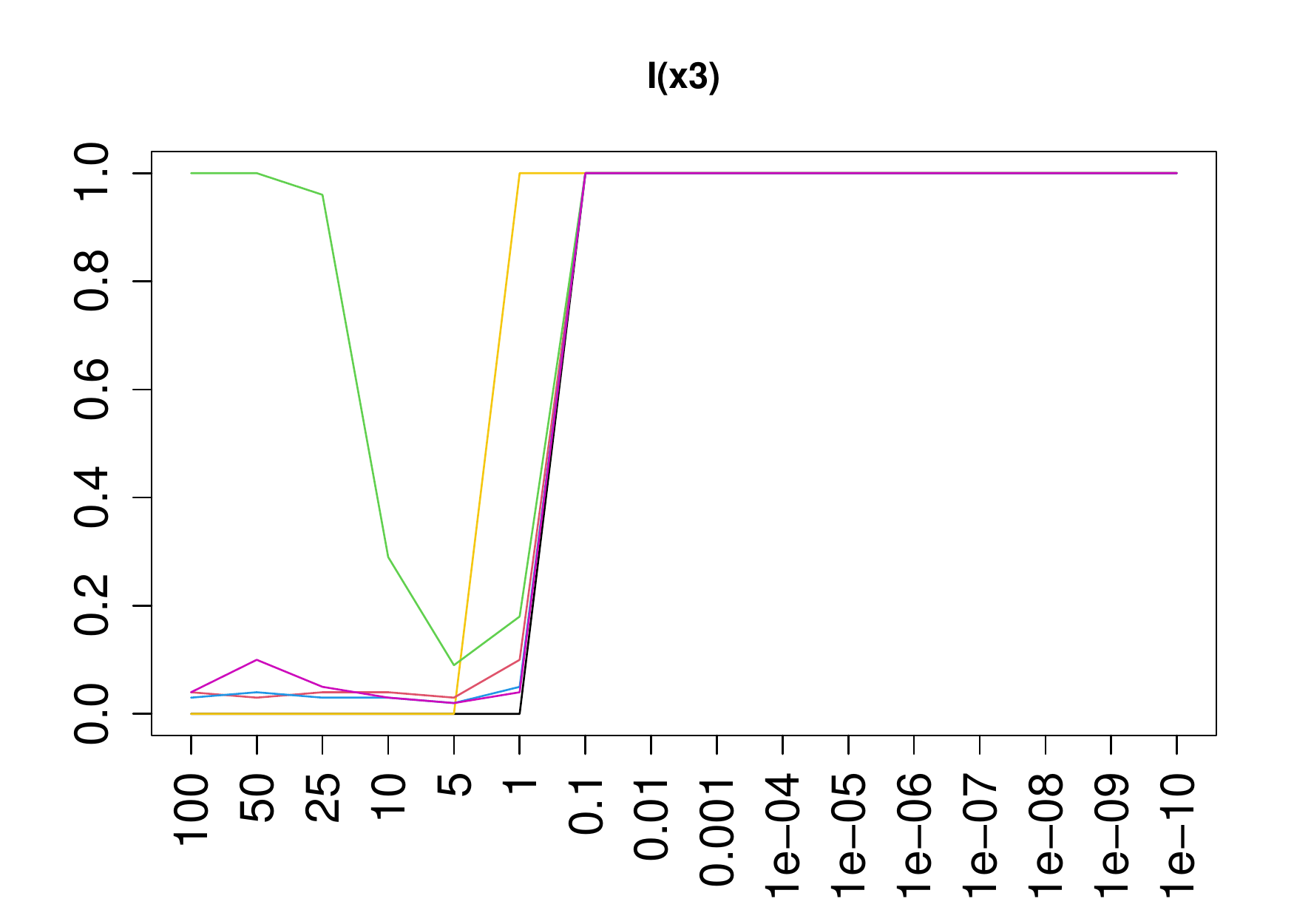}
    \includegraphics[trim={1cm 0 1cm 0},clip,scale=0.5]{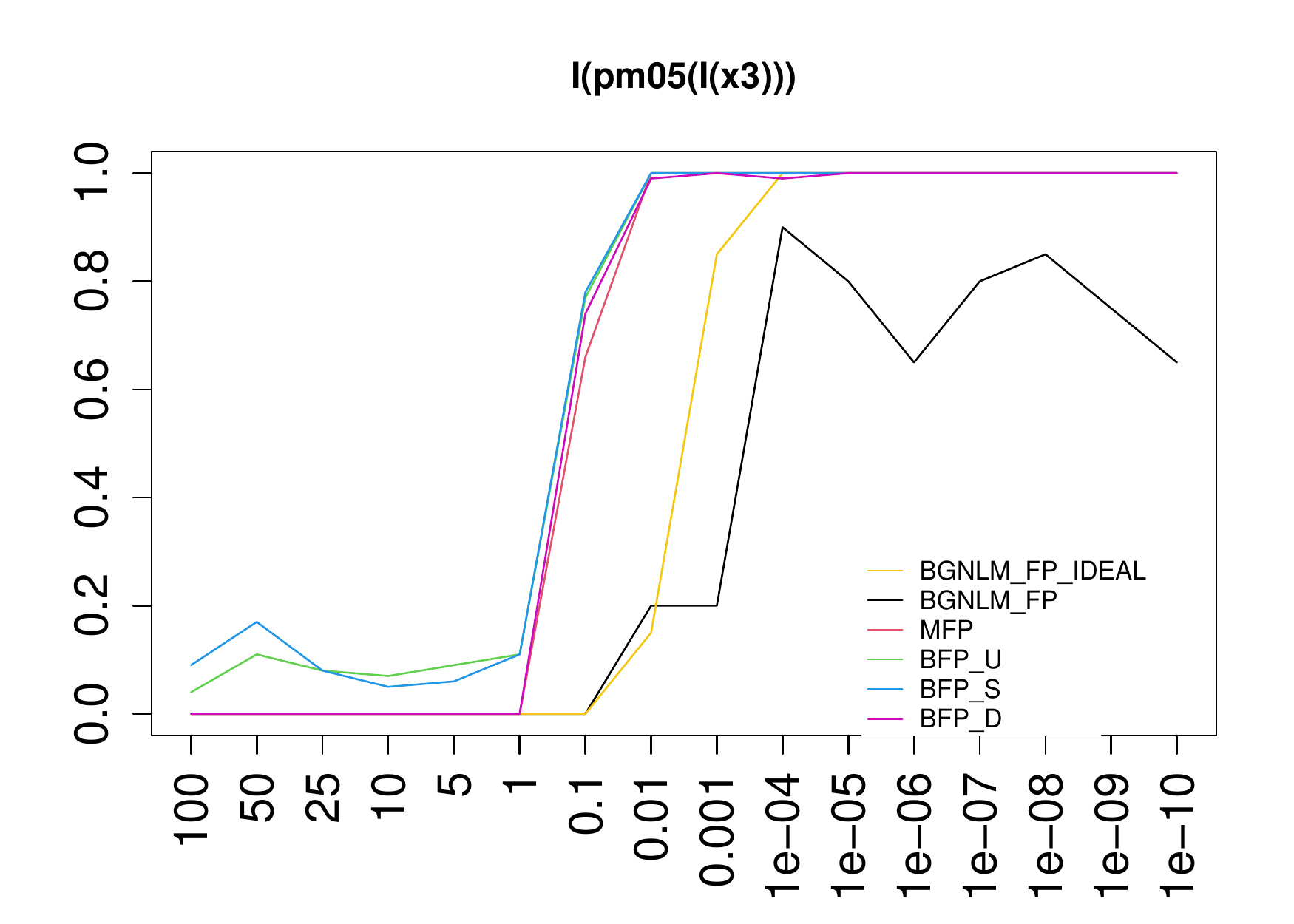}
    \caption{\label{fig:x3} Left panel: TPR for $x_3$ at the variable level for all approaches:  BGNLM\_FP (black), MFP (red), BFP with ``flat'' (BFP\_F, green), ``sparse'' (BFP\_S, blue), and ``dependent'' (BFP\_D, purple) priors. The yellow line shows BGNLM\_FP if the true model was forced into its search. Right panel: Same plot for $x_3$ at the FP(1) level, i.e., considering ``true positive'' the functional form FP1($x_3$, -0.5).}
\end{figure}

More interestingly, MFP and all the BFP models seem to almost always include $X_3$ as FP1(-0.5), meaning that they identify the first-order part of the transformation. BGNLM\_FP does not have the same behavior, and even at the largest signal-to-noise ratio sometimes stuck to correlated form (most probably FP1(-1), see the right panel of Figure \ref{fig:x3}). None of the approaches (except for the artificial BGNLM\_BFP\_IDEAL), anyhow, manages to identify the correct FP2(-0.5, -0.5) form (data not shown), no matter how large the signal-to-noise ratio is. Only the version of BGNLM\_FP\_IDEAL with the right model forced into the search path does, starting at $\sigma^2 = 0.01$ and fully happening from $\sigma^2 = 0.0001$. It can be inferred by contrasting the right plot of Figure \ref{fig:x3} and the top left panel of Figure \ref{fig:pmeans}.

\section{Real data applications}\label{sec:real}

In this section, we contrast our approach with many competitors in real data examples with responses of different natures, namely a continuous response, a binary response, and a time-to-event response. Note that the Bayesian approaches  considered here are based on sampling from the posterior, and therefore contain a stochastic component. For these algorithms, we perform 100 runs and report the median, the minimum, and the maximum result. As a consequence, we can also evaluate their stability.

\subsection{Regression task on the Abalone shell dataset}

The Abalone dataset, publicly available at \url{https://archive.ics.uci.edu/ml/datasets/Abalone}, has served as a reference dataset for prediction models for more than two decades. The goal is to predict the age of the abalone from physical measurements such as gender, length, diameter, height, whole weight, peeled weight, the weight of internal organs, and the shell. The response variable, age in years, is obtained by adding 1.5 to the number of rings. There are a total of $4177$ observations in this dataset, of which $3177$ were used for training and the remaining $1000$ for testing.

To compare all approaches, we use the following metrics:
\begin{itemize}
  \item Root Mean Square Error: $\text{RMSE} = \sqrt{\frac{\sum_{i=1}^{n_\text{test}}(\hat y_i -y_i)^2}{n_\text{test}}}$;
  \item Mean Absolute Error: $\text{MAE} = {\frac{\sum_{i=1}^{n_\text{test}}|\hat y_i-y_i|}{n_\text{test}}}$;
  \item Pearson's correlation between observed and predicted response: $\text{CORR}={\frac{\sum_{i=1}^{n_\text{test}}(\hat y_i^-\Bar{\hat y})( y_i-\Bar{y})}{\sqrt{\sum_{i=1}^{n_\text{test}}(\hat y_i-\Bar{\hat y})^2}\sqrt{\sum_{i=1}^{n_\text{test}}( y_i-\Bar{y})^2}}}$.
\end{itemize}
where $n_\text{test}$ here equal to $1000$, is the test sample size. In addition to the aforementioned approaches, here we also include the original BGNLM (see Equation \eqref{eq:BGNLM}) from \citet{hubin2021flexible} and a version with only linear terms BGLM (see Equation \eqref{eq:LRM}, with $\rho(x) = x$). 

\begin{table}[!ht]{
\resizebox{\textwidth}{!}{%
\begin{tabular}{llll}%
\hline 
Model&RMSE&MAE&CORR\\\hline
BGNLM&1.9573 (1.9334, 1.9903) & 1.4467 (1.4221, 1.4750) & 0.7831 (0.7740, 0.7895)\\
BFP\_F& 1.9649  (1.9649, 1.9649) & 1.4617  (1.4617, 1.4617) & 0.7804  (0.7804, 0.7804)\\
BFP\_S& 1.9649  (1.9649, 1.9649) & 1.4617  (1.4617, 1.4617) & 0.7804  (0.7804, 0.7804)\\
BGNLM\_FP & 1.9741 (1.9649, 2.0056) & 1.4679 (1.4623, 1.4896) & 0.7783 (0.7702, 0.7805)\\
MFP & 1.9792 (-, -) & 1.4710 (-, -) & 0.7770 (-, -)\\
BFP\_D & 1.9754 (1.9754, 1.9769) & 1.4668 (1.4668, 1.4677) & 0.7779 (0.7774, 0.7779)\\
BGLM & 2.0758 (2.0758, 2.0758) & 1.5381 (2.0758, 2.0758) & 0.7522 (2.0758, 2.0758)\\
\hline
\end{tabular}}
} 
\caption{\label{tAbalone} Abalone shell dataset: Prediction performances in terms of Root Mean Square Error (RMSE), Mean Absolute Error (MAE), and Pearson's correlation between observed and predicted response (CORR) for the different models. For methods with variable outcomes, the median measures (with minimum and maximum in parentheses) are displayed. The methods are sorted according to the median RMSE.} 
\end{table}

For BGNLM\_FP, GMJMCMC was run on $32$ parallel threads for each of the $100$ seeds. Each thread was run until $10,000$ unique models were visited, with a mutation rate of $250$ and the last mutation at iteration $10,000$. The population size of the GMJMCMC algorithm was set to $15$. For all runs, we had the following hyper-parameters of the model priors:  $q=15$ and $d=15$. Further, $a_k$ was chosen to be $a_k = \exp(-\log n)$ for $k: \rho_k \in \mathbf{F}_0$, $a_k = \exp(- (1+\log 2)\log n)$ for $k: \rho_k \in \mathbf{F}_1$, and $a_k = \exp(- (1+\log 4)\log n)$ for $k: \rho_k \in \mathbf{F}_2$. 

The best performance (see Table \ref{tAbalone}) is obtained with the general BGNLM approach. This is probably not surprising, as the relationship between response and explanatory variables seems complex (see also Table \ref{tAbalone2}) and BGNLM is the most flexible approach, which contains all the other models as special cases. Notably, this result seems to show that the GMJMCMC algorithm is effective in exploring the huge model space. On the other hand, the performance of BGLM, ranking the worst in all the three metrics considered, shows the importance of including non-linear effects in the model when analyzing this data set.

Between these two extremes lie all the FP implementations. Our proposed approach BGNLM\_FP seems slightly better than MFP and BFP\_D but worse than the other two implementations of BFP (BFP\_F and BFP\_S), which, in this case, have exactly the same performances. Nonetheless, no matter which metrics we consider, the differences among all FP-based approaches are very small.

Table \ref{tAbalone2} provides insight into the variable selection for our approach. This helps us to identify non-linear effect and give us a hint of the variable importance for the prediction task. The frequency of inclusion shows that all 9 explanatory variables were selected in all 100 simulation runs, meaning that each variable is relevant, at least with a linear effect. In addition, many non-linear effects had a posterior probability larger than 0.1 (see the right column of Table \ref{tAbalone2}). In particular, the variables WholeWeight, ShuckedWeights, Height, Length, and VisceraWeights seem to have an effect between quadratic and cubic, while Shell Weight seems to have a logarithmic effect, as the logarithmic transformation is selected 58\% of the times (third row of Table \ref{tAbalone2}), against around 20\% of other transformations (quadratic 21\%, cubic 16\%, and $x^{-0.5}$ 13\%). The presence of these non-linear polynomials in the model indicates that the relationship between the explanatory variables and the response (abalone age) is most probably non-linear and highlights the importance of using methods like BGNLM\_FP to predict the outcome.

\begin{table}[!ht]{
\begin{center}
\begin{tabular}{lc|lc}%
Liner Effects &Frequency&Non linear Effects&Frequency\\
\hline 
ShuckedWeight & 100 & WholeWeight (p = 2) & 68\\
Male & 100 & ShuckedWeight (p = 2) & 59\\
Diameter & 100 & ShellWeight (p = 0) & 58\\
Length & 100 & ShuckedWeight (p = 3) & 46\\
WholeWeight & 100 & Height (p = 2) &43\\
Height & 100 & Length (p = 3) & 39\\
VisceraWeight & 100 & VisceraWeight (p = 3) & 32\\
ShellWeight & 100 & VisceraWeight (p = 2) & 31\\
Femele & 100 & Height (p = 3) & 30\\
&& Length (p = 2) & 29\\
&& WholeWeight (p = 3) & 28\\
&& ShellWeight (p = 2) & 21\\
&& Height (p = 0) & 20\\
&& ShellWeight (p = 3) & 16\\
&& ShellWeight (p = -0.5) & 13\\
&& ShuckedWeight (p = 0) & 10
\end{tabular}
\end{center}
} 
\caption{\label{tAbalone2} Abalone shell dataset: Frequency of selection of the 9 explanatory variables and all nonlinear transformations with a posterior inclusion probability above 0.1 in more than 10 out of 100 simulation runs (in brackets the ``power'' of the transformation. The frequency is the number of simulations that include the given feature, linear features are listed on the left, and nonlinear on the right.}   
\end{table}

\subsection{Classification task on the Wisconsin breast cancer dataset}

The example uses breast cancer data with $357$ benign and $212$ malignant tissue observations, which were obtained from digitized fine needle aspiration images of a breast mass. The data can be found at the website \url{https://archive.ics.uci.edu/ml/datasets/Breast+Cancer+Wisconsin+(Diagnostic)}. Each cell nucleus is described by $10$ characteristics, including radius, texture, perimeter, area, smoothness, compactness, concavity, points of concavity, symmetry, and fractal dimension. For each variable, the mean, standard error, and mean of the top $3$ values per image were calculated, resulting in $30$ explanatory variables per image. The study used a randomly selected quarter of the images as the training data set, and the rest of the images were used as the test set.

As in the previous example, we compare the performance of the BGNLM\_FP to that of other methods, namely MFP, BGNLM, and its linear version BGLM. As BFP is not available for classification tasks, it could not be included in the comparison. BGNLM\_FP uses Bernoulli observations and a logit link function, and the variance parameter is fixed at $\phi=1$. Forecasts are made using $\hat{y}_i = \text{I}\left(\hat p(Y_i=1)\ge 0.5 \right)$, where $Y_i$ represents the response variable in the test set. The model averaging approach is used for prediction, where marginal probabilities are calculated using the Laplace approximation.

For BGNLM\_FP, GMJMCMC was run on $32$ parallel threads for each of the $100$ seeds. Each thread was run until $10,000$ unique models were visited, with a mutation rate of $250$ and the last mutation at iteration $10,000$. The population size of the GMJMCMC algorithm was set to $45$. For all runs, we had the following hyper-parameters of the model priors:  $q=45$ and $d=16$. $a_k$ was chosen to be as follows: $a_k = \exp(-\log n)$ for $k: \rho_k \in \mathbf{F}_0$, $a_k = \exp(- (1+\log 2)\log n)$ for $k: \rho_k \in \mathbf{F}_1$, and $a_k = \exp(- (1+\log 4)\log n)$ for $k: \rho_k \in \mathbf{F}_2$. 

To evaluate the performance of the models we computed the following metrics:
\begin{itemize}
    \item Prediction Accuracy: $\text{ACC} = \frac{\sum_{i=1}^{n_\text{test}}\text{I}(\hat y_i=y_i) }{n_\text{test}}$;
    \item False Positive Rate: \text{FPR} = $\frac{\sum_{i=1}^{n_\text{test}} \text{I}\left(y_i=0,\hat y_i=1\right)}{\sum_{i=1}^{n_\text{test}}  \text{I}\left(y_i = 0\right)}$;
    \item False Negative Rate: $\text{FNR} = \frac{\sum_{i=1}^{n_\text{test}}  \text{I}\left(y_i=1,\hat y_i=0\right)}{\sum_{i=1}^{n_\text{test}} \text{I}\left(y_i=1\right)}$.
\end{itemize}

The choice of these metrics is in line with that of \cite{hubin2020novel}, and allows direct comparison with the results therein.

\begin{table}[t]{
\resizebox{\textwidth}{!}{%
\begin{tabular}{llll}%
\hline 
Model&ACC&FNR&FPR\\\hline
BGNLM & 0.9742 (0.9695,0.9812) & 0.0479 (0.0479,0.0536) & 0.0111 (0.0000,0.0184)\\
BGLM & 0.9718 (0.9648,0.9765) & 0.0592 (0.0536,0.0702) & 0.0074 (0.0000,0.0148)\\
BGNLM\_FP & 0.9601 (0.9554,0.9648) & 0.0756 (0.0702,0.0809) & 0.0756(0.0702,0.0809)\\
MFP & 0.9413 (-,-) & 0.1011 (-,-) & 0.0255 (-,-)\\\hline
\end{tabular}}
}
\caption{\label{t2} Breast cancer classification dataset: Prediction performances in terms of Prediction Accuracy (ACC),  False Negative Rate (FNR), and False Positive Rate (FPR) for the different models. For methods with variable outcomes, the median measures (with minimum and maximum in parentheses) are displayed. The models are sorted according to median accuracy.} 
\end{table}

Table \ref{t2} presents the results for each metric. We can see that BGNLM\_FP performs better than MFP both in terms of prediction accuracy and false negative rate, while it is slightly worse than MFP when it concerns the false positive rate. Both FP-based models, however, perform worse than both BGNLM and its linear version BGLM. The very good performance of the latter, almost as good as the former in terms of accuracy and FNR, and even slightly better in terms of FPR, seems to suggest that non-linearities are not very important for this classification problem. This also explains why there is not much advantage in using an FP-based method. Both the frequentist and our proposed procedures tend to only select linear effects, as can be seen (for BGNLM\_FP) from Table \ref{tbc2}, where all the effects selected more than 10 (out of 100) runs are reported. The same happens for BGNLM \citep[see][Table 4]{hubin2021flexible}: even the most general model only selects mostly linear effects. The reason why BGNLM and BGLM have better results in this example is most probably related to better use of the priors (see the Discussion for more on this point).

\begin{table}[!ht]{
\begin{center}
\begin{tabular}{lc|lc}%
Effects &Frequency&Effects&Frequency\\
\hline 
fractal\_dimension\_se&100&radius\_worst&100\\
smoothness\_se&100&symmetry\_mean&100\\
concave.points\_se&100&texture\_mean&100\\
fractal\_dimension\_worst&100&compactness\_se&100\\
concavity\_mean&100&compactness\_worst&100\\
area\_worst&100&texture\_worst&100\\
smoothness\_mean&100&concavity\_worst&100\\
perimeter\_mean&100&perimeter\_se&100\\
compactness\_mean&100&concavity\_se&100\\
concave.points\_worst&100&symmetry\_worst&100\\
perimeter\_worst&100&area\_se&100\\
texture\_se&100&radius\_se&100\\
smoothness\_worst&100&fractal\_dimension\_mean&100\\
symmetry\_se&100&area\_mean&100\\
radius\_mean&100&concave.points\_mean&97\\
\hline
\end{tabular}
\end{center}
} 
\caption{\label{tbc2}  Breast cancer classification dataset: Frequency of selection of the explanatory variables with a posterior inclusion probability above 0.1 in more than 10 out of 100 simulation runs.}   
\end{table}

%In the breast cancer data, it can be seen from Table \ref{tbc2} that the most frequently selected features are all linear features and there is no occurrence of any nonlinear fractional polynomial terms. .  Thus, linear Bayesian regression performs on par with BGNLM. On the other hand, the performance of FFP and BFP may be worse because they may overfit the data and find non-relevant nonlinearities. However, unlike FFP and BFP, BGNLM\_BFP does not overfit the data due to its different priors, resulting in a more conservative regularisation inherited from BGNLM. This helps BGNLM\_BFP to avoid overfitting and ensures its good performance on this data.

\subsection{Time-to-event analysis on the German breast cancer study group dataset}

As an example outside the GLM context, we consider a dataset with a time-to-event response. In particular, the German Breast Cancer Study Group dataset contains data from 686 patients with primary node-positive breast cancer enrolled in a study from July 1984 to December 1989. Out of the 686 patients, 299 experience the event of interest (death or cancer recurrence), while the remaining 387 are censored observations. The data are publicly available at \url{https://www.uniklinik-freiburg.de/fileadmin/mediapool/08_institute/biometrie-statistik/Dateien/Studium_und_Lehre/Lehrbuecher/Multivariable_Model-building/gbsg_br_ca.zip} and contains information about 8 variables: 5 continuous (age, tumor size, number of positive nodes, progesterone status, and estrogen status), 2 binary (menopausal status and hormonal treatment) and one ordinal variable with 3 stages (tumor grade). The training set contains about 2/3 of the observations (457), with the remaining 1/3 forming the test set. The observations are randomly split, but the proportion of censored observations is forced to be the same in the two sets.

As in the previous examples, here we compare our approach BGNLM\_FP with a few competitors, namely the general BGNLM, its linear version BGLM, the classical MFP, and a linear version of the latter as well. All approaches are based on the partial likelihood of Equation \eqref{CoxL}, so all approaches provide a Cox model, with the latter model being the simple Cox regression model. Also, in this case, BFP is out of the game as it is only developed for Gaussian responses.

For BGNLM\_FP, GMJMCMC was run on $32$ parallel threads for each of the $100$ seeds. Each thread was run for $20,000$ iterations, with a mutation rate of $250$ and the last mutation at iteration $15,000$. The population size of the GMJMCMC algorithm was set to $15$. For all runs, we had the same as in all other examples hyper-parameters of the model priors:  $q=15$ and $d=15$. Further, $a_k$ was chosen to be $a_k = \exp(-\log n)$ for $k: \rho_k \in \mathbf{F}_0$, $a_k = \exp(- (1+\log 2)\log n)$ for $k: \rho_k \in \mathbf{F}_1$, and $a_k = \exp(- (1+\log 4)\log n)$ for $k: \rho_k \in \mathbf{F}_2$.

To evaluate the performance of the models, here we compute the following metrics:
\begin{itemize}
    \item Integrated Brier Score: $\text{IBS} = \int_0^{\max\{t_\text{test}\}} \frac{1}{n_\text{test}}\sum_{i=1}^{n_\text{test}} \hat{W}_i(t) \left(I(T_i > t) - \hat{S}_i(t) \right)^2 \, dt$;
    \item Concordance index: \text{C-index} = $\frac{ \sum_{i, j = 1}^{n_\text{test}} I(T_j < T_i) I(\mu(X_j) > \mu(X_i))\tilde{W}_i(t)\delta_j}{\sum_{i, j = 1}^{n_\text{test}} I(T_j < T_i) \tilde{W}_i(t)\delta_j}$;
\end{itemize}
where $\hat{W}_i(t)$ and $\tilde{W}_i(t)$ are weights, based on the Inverse Probability of Censoring Weighting \citep{ChengAl1995}, needed to take into account the presence of censored observations, $S_i$ is the survival function and $\delta_i$ is the censoring status for the patient $i$. Both IBS and C-index are computed through the R package \texttt{pec} \citep{Gerds2022}.

%\begin{table}[!ht]{
%\resizebox{\textwidth}{!}{%
%\begin{tabular}{llll}%
%\hline 
%Model&CINDEX&IBS1&IBS2\\\hline
%{BGNLM\_BFP{\_}PRL} & 0.6025 (0.5306,0.6248) & 0.1513 (0.1499,0.1537) & 0.2009 (0.1447,0.2631)\\
%BFP{\_}PRL (AIC) & 1.9860 (1.9667,2.0338) & 1.4746	(1.4633,1.5099)	& 0.7752 (0.7640,0.7780)\\
%{BGNLM\_BFP{\_}PRL{\_}I} & 0.5882 (0.4959,0.6444) & 0.1508 (0.1495 ,0.1543) & 0.1759 (0.1431,0.3403)\\
%{BGNLM\_BFP} & 0.5827 (0.4579,0.6396) & 0.1513 (0.1504,0.1535) & 0.1581 (0.1380,0.3909)\\
%BFP{\_}PRL (AIC) & 1.9860 (1.9667,2.0338) & 1.4746	(1.4633,1.5099)	& 0.7752 (0.7640,0.7780)\\
%{BGNLM\_BFP{\_}I} & 0.5280 (0.4175,0.6661) & 0.1518 (0.1500,0.1541) & 0.1669 (0.1408,0.4709)\\
%FFP& 0.5164 (-,-) & 0.1520 (-,-) & 0.1704 (-,-)\\
%LR & 0.4802 (-,-)& 0.1520 (-,-) & 0.1840 (-,-)\\
%\hline
%\end{tabular}}
%} 
%\caption{\label{tsurv3} Results on survival data.} 
%\end{table}

\begin{table}[!ht]
\begin{center}
\begin{tabular}{lll}%
\hline 
Model&IBS&CINDEX\\\hline
MFP  & 0.1609 (-,-) & 0.6939 (-,-) \\
BGNLM\_FP & 0.1619 (0.1604,0.1635) & 0.6913 (0.6871,0.6960) \\
%{BGNLM\_BFP{\_}PRL{\_}I} & 0.6885 (0.6663,0.6973) & 0.1895 (0.1886,0.1898) & 0.1623 (0.1597,0.1660)\\
% {BGNLM\_BFP} & 0.6821 (0.6141,0.7083) & 0.1891 (0.1882,0.1897) & 0.1638 (0.1570,0.1775)\\
%BFP{\_}PRL (AIC) & 1.9860 (1.9667,2.0338) & 1.4746	(1.4633,1.5099)	& 0.7752 (0.7640,0.7780)\\
%{BGNLM\_BFP{\_}I} & 0.6717 (0.5989,0.7129) & 0.1893 (0.1883,0.1902) & 0.1663 (0.1554,0.1771)\\
BGNLM & 0.1677 (0.1647,0.1792) & 0.6656 (0.6319,0.6801) \\
%BGNLM&  0.6591 (0.6100,0.7128) & 0.1891 (0.1877,0.1899) & 0.1688 (0.1555,0.1813)\\
BGLM & 0.1697 (0.1697,0.1697) & 0.6497 (0.6494,0.6500) \\
linear & 0.1701 (-,-) & 0.6184 (-,-) \\
\hline
null model & 0.1893 & 0.504 \\
\hline
\end{tabular}
\end{center} 
\caption{\label{tsurv31} German Breast Cancer Study Group dataset: Prediction performances in terms of Integrated Brier Score (IBS) and concordance index (C-index) for the different models. For methods with variable outcomes, the median measures (with minimum and maximum in parentheses) are displayed. The models are sorted according to the median IBS.} 
\end{table}

Table \ref{tsurv31} reports the results of this experiment. This dataset was used by \cite{RoystonSauerbrei2008} to illustrate the fractional polynomials, so probably not surprisingly the two FP-based approaches have the best performance. Both MFP and our proposed BGNLM\_FP are better than the competitors, especially those based on linear effect. It is known, indeed, that the effect of the variable nodes is not linear \citep[][Section 3.6.2]{RoystonSauerbrei2008}, and our approach finds this non-linearity 100\% of the times (see Table \ref{tsurvl2}). A bit more surprisingly BGNLM does not perform so well in this example, but it is most probably related to the fact that the extreme simplicity of a good model (at least the one found by BGNLM\_FP only contains two explanatory variables, one of them even with a simple linear effect) does not justify the use of complex machinery.

\begin{table}[!ht]{
\begin{center}
\begin{tabular}{lc|lc}%
Liner Effect & Frequency & Non linear Effect & Frequency\\
\hline 
progesterone status & 100 & FP1(number of positive nodes, 0) & 100\\
\hline
\end{tabular}
\end{center}
} 
\caption{\label{tsurvl2} German Breast Cancer Study Group dataset: Frequency of selection of the variables / nonlinear transformations with a posterior inclusion probability above 0.1 in more than 10 out of 100 simulation runs. Linear on the left, nonlinear on the right.}  \end{table}

\subsection{Including interaction terms into the models}

As discussed in Section \ref{interactions},  our approach makes it straightforward to add interaction terms in the Bayesian fractional polynomials models. Mathematically, we need to go back to formula \eqref{eq:BGNLM} and also consider bivariate transformation, while algorithmically we need to enable multiplication operators in the GMJMCMC algorithm. In this section, we report the results obtained by BGNLM\_FP with interactions. We kept all other tuning parameters of GMJMCMC unchanged, except for allowing multiplications. Also, all hyperparameters of the models are unchanged except for setting $d = \infty$ and $I = 4$.

\paragraph{Abalone data.} As we can see in Table \ref{tAbalone3}, allowing interactions into the model enhances the performance of the BGNLM\_BFP model on the Abalone shell age dataset. Adding the interactions is not sufficient to reach the performances of the general BGNLM, but it considerably reduces the gap.

\begin{table}[!ht]{
\resizebox{\textwidth}{!}{%
\begin{tabular}{llll}%
\hline 
Model&RMSE&MAE&CORR\\\hline
BGNLM&1.9573 (1.9334,1.9903)&1.4467 (1.4221,1.4750)&0.7831 (0.7740,0.7895)\\
BGNLM\_FP\_withInteraction & 1.9660 (1.9397,2.0039)&1.4514	(1.4326,1.4759)&0.7812	(0.7705,0.7874)\\
BGNLM\_FP & 1.9741 (1.9649,2.0056)& 1.4679	(1.4623,1.4896) & 0.7783 (0.7702,0.7805)\\
\hline
\end{tabular}}
} 
\caption{\label{tAbalone3} Abalone shell dataset: results for the BGNLM\_FP model when allowing for interactions. The results for BGNLM and BGNLM\_FP are reported form Table \ref{tAbalone} for comparison.} 
\end{table}

\paragraph{Breast cancer classification data.} As expected from the results of Table \ref{t2}, in the case of the breast cancer classification dataset, in which BGLM already performs better than BGNLM\_FP, incorporating interaction terms in the FP model does not produce any substantial advantage (see Table \ref{t23}). This does not come as a surprise, as nonlinearities do not seem to play a credible role in the prediction model.

\begin{table}[!ht]{
\resizebox{\textwidth}{!}{%
\begin{tabular}{llll}%
\hline 
Model&ACC&FNR&FPR\\\hline
BGNLM & 0.9742 (0.9695,0.9812) & 0.0479 (0.0479,0.0536)&0.0111 (0.0000,0.0184)\\
BGNLM\_FP\_withInteractions &0.9601	(0.9554,0.9671)& 0.0702	(0.0647,0.0809)& 0.0702 (0.0647,0.0809)\\
BGNLM\_FP & 0.9601 (0.9554,0.9648) & 0.0756 (0.0702,0.0809) & 0.0756 (0.0702,0.0809)\\
\hline
\end{tabular}}
}
\caption{\label{t23} Breast Cancer classification dataset: results for the BGNLM\_FP model when allowing for interactions. The results for BGNLM and BGNLM\_FP are reported form Table \ref{t2} for comparison.}   
\end{table}

\paragraph{German Breast Cancer Study Group data.}  Finally, Table \ref{tsurv312} shows the result of the model with interaction for the time-to-event data analysis. There is no advantage in allowing for interactions here as well. This is a typical case of the advantage related to the bias-variance trade-off when using simpler models for prediction tasks. We can notice from its IBS values that the model with interaction can have the best performance (0.1597), but the performances vary so much (as bad as 0.1660) that in the median is worse than the simpler model (that without interactions). Note, moreover, that Table \ref{tsurvl2} seems to suggest that there are only two relevant variables, so it is not likely to have relevant interactions.

\begin{table}[!ht]
\begin{center}
\begin{tabular}{lll}%
\hline 
Model& IBS & CINDEX \\\hline
BGNLM\_FP & 0.1619 (0.1604,0.1635) & 0.6913 (0.6871,0.6960) \\
BGNLM\_FP\_withInteractions & 0.1623 (0.1597,0.1660) & 0.6885 (0.6663,0.6973) \\
BGNLM & 0.1677 (0.1647,0.1792) & 0.6656 (0.6319,0.6801) \\
\hline
\end{tabular}
\end{center}
\caption{\label{tsurv312} German Breast Cancer Study Group dataset: results for the BGNLM\_FP model when allowing for interactions. The results for BGNLM and BGNLM\_FP are reported from Table \ref{t2} for comparison.} 
\end{table}

\section{Discussion}\label{sec:conclusions}

In this paper, we studied how BGNLM fitted by GMJMCMC introduced by \citet{hubin2021flexible} can deal with fractional polynomials. It can be seen as an opportunity of fitting a BGNLM that can handle nonlinearities without any loss in model interpretability, and, more importantly, as a convenient implementation of a fractional polynomial model that assures a coherent inferential framework, without losing (if not gaining) anything in terms of prediction ability. The broad generality of the BGNLM framework, moreover, allows adding complexity with a minimum effort, as we show for the inclusion of interaction terms.

%The approach was then compared to the existing implementation of MFP by \citet{HeinzeAl2022} and the implementation of BFP by \citet{SabanesboveaL2022}. The simulation study shows promising results of BGNLM\_FP. It outperforms BFP in terms of a strict definition of Power and FDR. Also, the performance is on par with that of FFP. 
Note that the current implementation is based on a direct adaptation of the priors defined in \citet{hubin2021flexible}. Further investigations on the choice of the priors will certainly be beneficial and further improve the performances of BGNLM\_FP, as we already noticed in the simulation study of Section \ref{sec:ART}. For example, a better balance between the penalty for the different fractional polynomial forms can be implemented. Our real data experiments never showed evidence in favor of a fractional polynomial of order 2. This may be related to the implausibility of a FP(2) transformation, especially in a prediction context where simplicity is often awarded, but it may also indicate that we penalized these terms too much.

%In regards to predictive tasks, the results suggest that BGNLM\_BFP performs similarly to both FFP and BFP in terms of RMSE, MAE, and correlation for  regression tasks. It is also  on par with FFP for classification and survival prediction tasks (standard BFP is not implemented for these types of problems). However, BGNLM surpasses all of these methods in terms of performance. The superior performance of BGNLM over BFP in predictive tasks can be attributed to its more comprehensive framework, which allows for not only nonlinear effects but also interactions. %The results demonstrate that BGNLM can deliver the best results in terms of accuracy, FNR, FPR, RMSE, MAE, and correlation when compared to the other methods.

One drawback with the Bayesian version of the fractional polynomial models is the computational costs of fitting it. This is not specific to our approach, it also concerns the current BFP implementation of \cite{SabanesboveHeld2011} and can become an issue in the case of very large datasets. Currently, we distribute the computational workload across multiple processors to achieve conference to descent regions of the model space. In the future, subsampling the data can allow for a reduction in computational cost when computing the marginal likelihoods. This will allow using the Bayesian fractional polynomial approach in big data problems as well. Furthermore, in the case of large datasets, Laplace approximations of the marginal likelihoods become very accurate, making this approach even more appealing. To make the computations efficient, stochastic gradient descent (SGD) can be used to compute the Laplace approximations, which also guarantees convergence \citep{lachmann2022subsampling} of MJMCMC in the class of Bayesian GLMs. Therefore, incorporating data subsampling and using SGD to compute the Laplace approximations may be a promising future direction in inference on Bayesian fractional polynomials under the settings of large $n$.

Another challenge is selecting appropriate values for the tuning parameters of GMJMCMC. The tuning parameters in GMJMCMC control the proposal distributions, population size, frequencies of genetic operators, and other characteristics of the Markov chain. Their values can significantly affect the convergence and mixing properties of the algorithm. To deal with this challenge, one may perform extensive tuning of the algorithm, which involves testing a range of values for the tuning parameters and evaluating the performance of the algorithm using problem-specific diagnostic tools like Power-FDR in simulations or RMSE for regression prediction tasks. A detailed discussion of setting tuning parameters of GMJMCMC is given in the Rejoinder \citep{hubin2020rejoinder} to the Discussion of \citet{hubin2020novel}. In the future, it may also be interesting to develop adaptive tuning methods that automatically adjust the tuning parameters based on the performance of the algorithm similarly to how it was done in \citet{hubin2019adaptive}.

GMJMCMC is a Markov chain Monte Carlo algorithm that is designed to explore the space of models with non-zero posterior probabilities. However, as with any MCMC algorithm, there is a risk that the chain may not converge to the desired target distribution in a finite time. This means in our settings that the set of models with non-zero posterior probabilities may not be fully explored in a single run of the algorithm.
One consequence of this is that the estimates obtained from GMJMCMC may vary from run to run, since different runs may explore different parts of the model space. Even if the algorithm is run for a long time, there is still a positive probability that it may miss some of the models with non-zero posterior probabilities. Variance in the estimated posterior in turn induces variance in the predictions if the latter is of interest. To mitigate this issue, it is recommended to run the algorithm multiple times using as many of the available resources as one can and check for convergence of the estimates. Additionally, it may be helpful to use informative model priors or other techniques to help guide the algorithm toward the most relevant parts of the model space.

Even though there are still a few challenges and limitations in the current state of Bayesian fractional polynomials, these models have potential applications in a variety of areas where uncertainty handling, explainability, and nonlinear relationships are essential. These include but are not limited to fields such as pharmacology, epidemiology, finance, and engineering. In pharmacology, for example, fractional polynomials can be used to model the dose-response relationship between a drug and a patient's response, taking into account the nonlinear and complex relationships between the variables. In finance, fractional polynomials can be used to model the relationship between financial variables such as stock prices, interest rates, and exchange rates, and to quantify the uncertainty associated with these relationships. Similarly, in engineering, fractional polynomials can be used to model the relationship between variables such as stress, strain, and material properties, providing a way to make predictions while taking into account nonlinear relationships and the associated uncertainty. In all of these cases, Bayesian fractional polynomials offer a flexible and robust way to handle uncertainty and model nonlinear relationships, making them a useful tool for a wide range of applications in the future. Given the important role that uncertainty handling, explainability, and nonlinear relationships play in various applications, we hope that the novel Bayesian fractional polynomials inference algorithm presented in this paper as well as the suggested extensions to various practical settings like survival analysis and GLM will allow these often overlooked models be more widely used in the future.

\bibliographystyle{hapalike}
\bibliography{biblio}

\begin{thebibliography}{}

\bibitem[Barbieri and Berger, 2004]{BarbieriBerger2004}
Barbieri, M.~M. and Berger, J.~O. (2004).
\newblock Optimal predictive model selection.
\newblock {\em Annals of Statistics}, 32:870--897.

\bibitem[Bayarri et~al., 2012]{bayarri2012criteria}
Bayarri, M.~J., Berger, J.~O., Forte, A., and Garc{\'\i}a-Donato, G. (2012).
\newblock Criteria for {B}ayesian model choice with application to variable
  selection.
\newblock {\em The Annals of Statistics}, 40:1550--1577.

\bibitem[Box and Tidwell, 1962]{BoxTidwell1962}
Box, G.~E. and Tidwell, P.~W. (1962).
\newblock Transformation of the independent variables.
\newblock {\em Technometrics}, 4:531--550.

\bibitem[Cheng et~al., 1995]{ChengAl1995}
Cheng, S., Wei, L.~J., and Ying, Z. (1995).
\newblock Analysis of transformation models with censored data.
\newblock {\em Biometrika}, 82:835--845.

\bibitem[Claeskens and Hjort, 2008]{claeskens_hjort_2008}
Claeskens, G. and Hjort, N.~L. (2008).
\newblock {\em Model Selection and Model Averaging}.
\newblock Cambridge Series in Statistical and Probabilistic Mathematics.
  Cambridge University Press.

\bibitem[Gelman et~al., 2013]{gelman2013bayesian}
Gelman, A., Stern, H.~S., Carlin, J.~B., Dunson, D.~B., Vehtari, A., and Rubin,
  D.~B. (2013).
\newblock {\em Bayesian data analysis}.
\newblock Chapman and Hall/CRC, Boca Raton.

\bibitem[Gerds, 2022]{Gerds2022}
Gerds, T.~A. (2022).
\newblock {\em pec: Prediction Error Curves for Risk Prediction Models in
  Survival Analysis}.
\newblock R package version 2022.05.04.

\bibitem[Heinze et~al., 2022]{HeinzeAl2022}
Heinze, G., Ambler, G., and Benner, A. (2022).
\newblock {\em mfp: Multivariable Fractional Polynomials}.
\newblock R package version 1.5.2.2.

\bibitem[Hubin, 2019]{hubin2019adaptive}
Hubin, A. (2019).
\newblock {An adaptive simulated annealing em algorithm for inference on
  non-homogeneous hidden Markov models}.
\newblock In {\em Proceedings of the international conference on artificial
  intelligence, information processing and cloud computing}, pages 1--9.

\bibitem[Hubin and Storvik, 2016]{HubinStorvikINLA}
Hubin, A. and Storvik, G. (2016).
\newblock {{E}stimating the marginal likelihood with {I}ntegrated nested
  {L}aplace approximation ({INLA})}.
\newblock 1611.01450.
\newblock arXiv:1611.01450v1.

\bibitem[Hubin and Storvik, 2018]{hubin2016efficient}
Hubin, A. and Storvik, G. (2018).
\newblock Mode jumping mcmc for {B}ayesian variable selection in {GLMM}.
\newblock {\em Computational Statistics \& Data Analysis}, 127:281--297.

\bibitem[Hubin et~al., 2020a]{hubin2020novel}
Hubin, A., Storvik, G., and Frommlet, F. (2020a).
\newblock A novel algorithmic approach to {B}ayesian logic regression (with
  discussion).
\newblock {\em Bayesian Analysis}, 15:263--333.

\bibitem[Hubin et~al., 2020b]{hubin2020rejoinder}
Hubin, A., Storvik, G., and Frommlet, F. (2020b).
\newblock Rejoinder for the discussion of the paper "{A Novel Algorithmic
  Approach to Bayesian Logic Regression}".
\newblock {\em Bayesian Analysis}, 15(1):312--333.

\bibitem[Hubin et~al., 2021]{hubin2021flexible}
Hubin, A., Storvik, G., and Frommlet, F. (2021).
\newblock Flexible {B}ayesian nonlinear model configuration.
\newblock {\em Journal of Artificial Intelligence Research}, 72:901--942.

\bibitem[Jeffreys, 1946]{jeffreys1946invariant}
Jeffreys, H. (1946).
\newblock An invariant form for the prior probability in estimation problems.
\newblock {\em Proceedings of the Royal Society of London A}, 186:453--461.

\bibitem[Lachmann et~al., 2022]{lachmann2022subsampling}
Lachmann, J., Storvik, G., Frommlet, F., and Hubin, A. (2022).
\newblock {A subsampling approach for Bayesian model selection}.
\newblock {\em International Journal of Approximate Reasoning}, 151:33--63.

\bibitem[Li and Clyde, 2018]{LiClyde2018}
Li, Y. and Clyde, M.~A. (2018).
\newblock Mixtures of g-priors in generalized linear models.
\newblock {\em Journal of the American Statistical Association},
  113:1828--1845.

\bibitem[Liang et~al., 2008]{LiangAl2008}
Liang, F., Paulo, R., Molina, G., Clyde, M.~A., and Berger, J.~O. (2008).
\newblock Mixtures of g priors for {B}ayesian variable selection.
\newblock {\em Journal of the American Statistical Association}, 103:410--423.

\bibitem[Raftery et~al., 1997]{raftery1997bayesian}
Raftery, A.~E., Madigan, D., and Hoeting, J.~A. (1997).
\newblock Bayesian model averaging for linear regression models.
\newblock {\em Journal of the American Statistical Association}, 92:179--191.

\bibitem[Raftery et~al., 2005]{RafteryAl2005}
Raftery, A.~E., Painter, I.~S., and Volinsky, C.~T. (2005).
\newblock {BMA}: an {R} package for {B}ayesian model averaging.
\newblock {\em The Newsletter of the R Project Volume}, 5:2--8.

\bibitem[Royston and Altman, 1994]{RoystonAltman1994}
Royston, P. and Altman, D.~G. (1994).
\newblock Regression using fractional polynomials of continuous covariates:
  parsimonious parametric modelling.
\newblock {\em Journal of the Royal Statistical Society: Series C (Applied
  Statistics)}, 43:429--453.

\bibitem[Royston and Sauerbrei, 2004]{RoystonSauerbrei2004}
Royston, P. and Sauerbrei, W. (2004).
\newblock A new measure of prognostic separation in survival data.
\newblock {\em Statistics in Medicine}, 23:723--748.

\bibitem[Royston and Sauerbrei, 2008]{RoystonSauerbrei2008}
Royston, P. and Sauerbrei, W. (2008).
\newblock {\em Multivariable Model-building: a pragmatic approach to regression
  anaylsis based on fractional polynomials for modelling continuous variables}.
\newblock Wiley, Chichester.

\bibitem[Rue et~al., 2009]{rue2009eINLA}
Rue, H., Martino, S., and Chopin, N. (2009).
\newblock {Approximate {B}ayesian inference for latent {G}aussian models by
  using integrated nested {L}aplace approximations}.
\newblock {\em Journal of the Royal Statistical Sosciety}, 71:319--392.

\bibitem[Saban{\'e}s~Bov{\'e} et~al., 2022]{SabanesboveaL2022}
Saban{\'e}s~Bov{\'e}, D., Gravestock, I., Davies, R., Moshier, S., Ambler, G.,
  and Benner, A. (2022).
\newblock {\em bfp: Bayesian Fractional Polynomials}.
\newblock R package version 0.0.46.

\bibitem[Saban{\'e}s~Bov{\'e} and Held, 2011]{SabanesboveHeld2011}
Saban{\'e}s~Bov{\'e}, D. and Held, L. (2011).
\newblock Bayesian fractional polynomials.
\newblock {\em Statistics and Computing}, 21:309--324.

\bibitem[Sauerbrei and Royston, 1999]{SauerbreiRoyston1999}
Sauerbrei, W. and Royston, P. (1999).
\newblock Building multivariable prognostic and diagnostic models:
  transformation of the predictors by using fractional polynomials.
\newblock {\em Journal of the Royal Statistical Society: Series A (Statistics
  in Society)}, 162:71--94.

\bibitem[Schmoor et~al., 1996]{SchmoorAl1996}
Schmoor, C., Olschewski, M., and Schumacher, M. (1996).
\newblock Randomized and non-randomized patients in clinical trials:
  experiences with comprehensive cohort studies.
\newblock {\em Statistics in Medicine}, 15:263--271.

\bibitem[Schwarz, 1978]{Schwarz}
Schwarz, G. (1978).
\newblock {Estimating the dimension of a model}.
\newblock {\em The Annals of Statistics}, 6:461--464.

\end{thebibliography}

\end{document}